\documentclass[11pt, reqno]{article}
\usepackage[hmargin=1.0in,vmargin=1.0in]{geometry}
\usepackage{amsthm}
\usepackage[english]{babel}
\usepackage{cite}
\usepackage{color}
\usepackage{amsmath,amsfonts,amssymb,graphicx}
\usepackage{microtype}
\usepackage{mathrsfs,url}
\usepackage{enumerate}
\usepackage{cmathml}
\usepackage{caption}
\usepackage{graphics}
\usepackage{epic}

\newtheorem{theorem}{Theorem}
\newtheorem{lemma}{Lemma} 
\newtheorem*{lemma*}{Lemma} 
\newtheorem*{proposition*}{Proposition} 
\newtheorem*{theorem*}{Theorem} 
\newtheorem{proposition}[theorem]{Proposition}
\newtheorem{corollary}{Corollary}
\newtheorem{definition}{Definition}
\theoremstyle{remark}
\newtheorem{example}{Example}

\newcommand{\ignore}[1]{}

\newcommand{\cpt}{\ensuremath{\operatorname{cpt}}}
\newcommand{\negskel}{\ensuremath{\operatorname{neg}}}

\newcommand{\tuple}[1]{\ensuremath{\langle #1 \rangle}}

\thicklines 
\newsavebox{\vartwo}
\savebox{\vartwo}(20,40){
\begin{picture}(20,40)(0,0)
\put(10,20){\oval(18,38)} \put(10,10){\makebox(0,0){$\bullet$}}
\put(10,30){\makebox(0,0){$\bullet$}}
\end{picture}
}
\newsavebox{\varone}
\savebox{\varone}(20,40){
\begin{picture}(20,40)(0,0)
\put(10,20){\oval(18,28)} \put(10,20){\makebox(0,0){$\bullet$}}
\end{picture}
}
\thinlines

\thicklines 
\newsavebox{\vartwobig}
\savebox{\vartwobig}(20,40){
\begin{picture}(20,40)(0,0)
\put(10,20){\oval(18,38)} \put(10,10){\makebox(0,0){$\bullet$}}
\put(10,30){\makebox(0,0){$\bullet$}}
\end{picture}
}
\newsavebox{\varonebig}
\savebox{\varonebig}(20,40){
\begin{picture}(20,40)(0,0)
\put(10,20){\oval(18,28)} \put(10,20){\makebox(0,0){$\bullet$}}
\end{picture}
}
\thinlines

\date{}

\begin{document}

\title{Variable and Value Elimination in Binary Constraint Satisfaction via
Forbidden Patterns\thanks{A preliminary version of part of this work appeared in
\emph{Proceedings of the 23rd International Joint Conference on Artificial
Intelligence (IJCAI)}, 2013. Martin Cooper and Guillaume Escamocher were
supported by ANR Project ANR-10-BLAN-0210. Stanislav \v{Z}ivn\'y was supported
by a Royal Society University Research Fellowship. David Cohen, Martin Cooper
and  Stanislav \v{Z}ivn\'y were supported by EPSRC grant  EP/L021226/1.}}

\author{
David A. Cohen\\
Department of Computer Science\\
Royal Holloway, University of London\\
UK\\
\texttt{dave@cs.rhul.ac.uk}
\and
Martin C. Cooper\\
IRIT, University of Toulouse III\\
31062 Toulouse\\
France\\
\texttt{cooper@irit.fr}
\and
Guillaume Escamocher\\
Insight Centre for Data Analytics\\
University College Cork\\
Ireland\\
\texttt{guillaume.escamocher@insight-centre.org}
\and
Stanislav \v{Z}ivn\'y\\
Department of Computer Science\\
University of Oxford\\
UK\\
\texttt{standa@cs.ox.ac.uk}
}

\maketitle

\begin{abstract}

Variable or value elimination in a constraint satisfaction problem
(CSP) can be used in preprocessing or during search to reduce search
space size. A variable elimination rule (value elimination rule)
allows the polynomial-time identification of certain variables
(domain elements) whose elimination, without the introduction of
extra compensatory constraints, does not affect the satisfiability of
an instance. We show that there are essentially just four variable
elimination rules and three value elimination rules defined by
forbidding generic sub-instances, known as irreducible existential
patterns, in arc-consistent CSP instances. One of the variable
elimination rules is the already-known Broken Triangle Property,
whereas the other three are novel. The three value elimination rules
can all be seen as strict generalisations of neighbourhood
substitution.

\end{abstract}


\section{Introduction}

Constraint satisfaction has proved to be a useful modelling tool in a
variety of contexts, such as scheduling, timetabling, planning,
bio-informatics and computer
vision~\cite{Dechter03:book,Rossi06:handbook,Lecoutre2009}. Dedicated
solvers for constraint satisfaction are at the heart of the
programming paradigm known as constraint programming. Theoretical
advances on CSPs can thus potentially lead to the improvement of
generic combinatorial problem solvers.

In the CSP model we have a number of variables, each of which can
take values from its particular finite domain. Certain sets of the
variables are constrained in that their simultaneous assignments of
values is limited. The generic problem in which these sets of
variables, known as the constraint scopes, are all of cardinality at
most two, is known as binary constraint satisfaction. We are required
to assign values to all variables so that every constraint is
satisfied. Complete solution algorithms for constraint satisfaction
are not polynomial time unless P=NP, since the graph colouring
problem, which is NP-complete, can be reduced to binary constraint
satisfaction~\cite{Dechter03:book}. Hence we need to find ways to
reduce the search space.

Search algorithms for constraint problems usually proceed by
transforming the instance into a set of subproblems, for example, by
selecting a variable and assigning to it successively each value from
its domain.  This naive backtracking approach is recursive and
explores the search tree of partial assignments in a depth first
manner. Even though the backtracking algorithm can take exponential
time it is often effective in practice thanks to intelligent pruning
techniques.

There are many ways to improve naive backtracking by pruning the
search space in ways that cannot remove solutions.  This is done by
avoiding searching exhaustively in all generated subproblems when
certain kinds of discovered obstruction to solution exists.  Such
techniques include Back-marking, Back-jumping, Conflict-Directed
Back-jumping ~\cite{Prosser93:hybrid,DBLP:journals/jair/ChenB01}. As
well as these look-back techniques it is also possible to look ahead
by propagating the consequences of early decisions or of the
discovered structure.  Of these look-ahead techniques the most common
is to maintain the local consistency property called generalised
arc-consistency (GAC)~\cite{Bessiere:AC}.  This technique identifies
certain values for variables that cannot possibly form part of a
solution.

Of course, savings can also be made if we are able to eliminate
variables from a sub-problem.  Since backtracking is of exponential
time complexity, the elimination of variables and values (domain
elements) to reduce instance size can in the best case reduce search
time by an exponential factor. To maintain the soundness of search we
require that such eliminations do not change the satisfiability of
the instance. Invariance of satisfiability, which we study in the
present paper, is a weaker property than the invariance of the set of
solutions guaranteed by consistency techniques such as GAC. However,
detection of non-satisfiability is the essential role of look-ahead
techniques, since this allows pruning during search. Thus
satisfiability-preserving reduction techniques (which do not
necessarily preserve solutions) may prove useful even when the aim is
to discover one or all solutions. In fact, we show that all the
techniques presented in this paper, although they do not preserve
solutions, allow a solution to the original instance to be
reconstructed very efficiently.

\subsection{Simplification by variable and value elimination}
\label{subsec:vvelim}
We consider an instance $I$ of the CSP viewed
as a decision problem. Suppose that $x$ is a variable of $I$ and
that, whenever there is some valid assignment to all variables except
$x$, there is a solution to the whole instance; in this case, we can
safely remove variable $x$ from $I$. One of the questions we address
in this paper is how to identify such variables?

Variable elimination has been considered before in the literature. It
is well known that in an arc-consistent binary CSP instance, a
variable $x$ which is constrained by only one other variable $y$ can be
eliminated; by the definition of arc consistency, each assignment to
$y$ is compatible with some assignment to $x$. It has been observed
that a more general property, called the (local) Broken Triangle
Property (lBTP)~\cite{Cooper10:BTP}, if it holds at some variable,
allows us to eliminate that variable. One way of stating the lBTP is
that there is no pair of compatible assignments to two other
variables $y,z$ which have opposite compatibilities with two
assignments to $x$. The closure of a binary CSP instance under the
elimination of all variables that satisfy the lBTP is unique and can
be found in $O(ncd^{3})$ time, where $n$ is the number of variables,
$c$ the number of constraints and $d$ the maximum domain size, which
may well prove effective when compared to the exponential cost of
backtracking.
The more general
local min-of-max extendable property (lMME) allows us to eliminate
more variables than the lBTP, but requires the identification of a
particular domain order. Unfortunately, this domain order is NP-hard
to discover~\cite{Cooper10:BTP} for unbounded domain size, and so the
lMME is less likely to be effective in practice.

An alternative to simple variable elimination is used in Bucket
Elimination~\cite{Larrosa03:boosting}.  In this algorithm a variable
$v$ is not simply eliminated.  Instead it is replaced by a constraint
on its neighbourhood (the set of variables constrained by $v$).  This
new constraint precisely captures those combinations of assignments
to the neighbourhood of $v$ which can be extended to a consistent
assignment to $v$. Such an approach may generate high-order
constraints, which are exponentially hard to process and to store.
The arity can be bounded by the induced treewidth of the instance,
but this still limits the applicability of Bucket Elimination. In the
present paper we restrict our attention to the identification of
variable elimination strategies which do not require the addition of
compensatory constraints.

The elimination of domain elements is an essential component of
constraint solvers via generalised arc consistency (GAC) operations.
GAC eliminates domain elements that cannot be part of any solution,
thus conserving all solutions.
An alternative approach is the family of elimination rules based on
substitution: if all solutions in which variable $v$ is assigned
value $b$ remain solutions when the value of variable $v$ is changed
to another value $a$, then the value $b$ can be eliminated from the
domain of variable $v$ while conserving at least one solution (if the
instance is satisfiable). The most well-known polynomial-time
detectable substitution operation is neighbourhood
substitution~\cite{Freuder91:interchangeable}. The value elimination
rules described in this paper go beyond the paradigms of consistency
and substitution; we only require that the instance obtained after
elimination of a value from a domain has the same satisfiability as
the original instance.

We study rules for simplifying binary CSP instances based on properties
of the instance at the microstructure level. The term microstructure was first
given a formal definition by J\'{e}gou~\cite{Jegou93:microstructure}: 
if $I$ is a binary CSP instance, then its
\emph{microstructure} is a graph $\tuple{A,E}$ where $A$ is the set of possible 
variable-value assignments and $E$ is the set of pairs of compatible
variable-value assignments. Solutions to $I$ are in one-to-one
correspondence with the $n$-cliques of the microstructure of $I$ and
with the size-$n$ independent sets of the microstructure complement
of $I$. The \emph{chromatic number} of a graph is the smallest number of
colours required to colour its vertices so that no two adjacent
vertices have the same colour. A graph $G$ is \emph{perfect} if for
every induced subgraph $H$ of $G$, the chromatic number of $H$ is
equal to the size of the largest clique contained in $H$. Since a
maximum clique in a perfect graph can be found in polynomial
time~\cite{Grotschel1981}, the class of binary CSP instances with a
perfect microstructure is tractable~\cite{Salamon2008:perfect}.
Perfect graphs can also be recognized in polynomial
time~\cite{Cornuejols03:perfect}. An instance of the minimum-cost homomorphism problem
(MinHom) is a CSP instance in which weights are associated with 
each variable-value assignment and the aim is to find a solution which
minimises the sum of the weights. Takhanov~\cite{Takhanov10:dichotomy} gave a dichotomy
for tractable conservative constraint languages for MinHom which
uses the fact that an instance of binary MinHom can be solved
in polynomial time if its microstructure is perfect.
El Mouelhi et al.~\cite{DBLP:conf/cpaior/MouelhiJTZ13} 
make the observation that if the microstructure has a bounded
number of maximal cliques then the instance will be solved in
polynomial time by classical algorithms such as Forward Checking
or Really Full Lookahead and hence by CSP solvers.

Simple rules for variable or value elimination based on properties of
the microstructure are used by Beigel and Eppstein~\cite{DBLP:conf/focs/BeigelE95} in their
algorithms with low worst-case time bounds for such NP-complete problems
as 3-COLOURING and 3SAT. Such simplification operations are an essential first step
before the use of decompositions into subproblems with smaller domains.
A similar approach allows Angelsmark and Thapper~\cite{DBLP:conf/flairs/AngelsmarkT05}
to reduce the problem of finding a minimum weighted independent set
in the microstructure complement to the problem of counting the number
of solutions to a 2SAT instance.
Thus, the variable and value elimination rules we present in this paper may find 
not only practical applications in solvers but also theoretical applications.

\subsection{Our contribution}
In this paper we characterise those local conditions under which we
can eliminate variables or values in binary CSPs while preserving
satisfiability of the instance, without the need to add compensating
constraints. By local conditions we mean configurations of variables,
values and constraints which \textit{do not} occur.  That is, we will
identify (local) obstructions to variable or value elimination.  We
will call such constructions variable elimination or value
elimination patterns.

Surprisingly we find that there are precisely four (three)
essentially different local patterns whose absence permits variable
(value) elimination. Searching for these local patterns takes
polynomial time and need only be done during the pre-processing
stage, before search. Any discovered obstructions to elimination can
be effectively monitored during subsequent search using techniques
analogous to watched literals~\cite{Gent06:watched}. Whenever a
variable (value) no longer participates in any obstruction patterns
it can safely be eliminated.

We show that after a sequence of variable and value eliminations it
is always possible to reconstruct a solution to the original instance
from a solution to the reduced instance in low-order polynomial time.

\section{Definitions}
When certain kinds of local obstructions are not present in a binary
CSP instance, variable or value elimination is possible. Such
obstructions are called quantified patterns. A pattern can be seen as
a generalisation of the concept of a constraint satisfaction instance
that leaves the consistency of some assignments to pairs of variables
undefined.

\begin{definition}
A \emph{pattern} is a four-tuple \tuple{X, D, A, \cpt} where:
\begin{itemize}
\item $X$ is a finite set of \emph{variables};
\item $D$ is a finite set of \emph{values};
\item $A \subseteq X \times D$ is the set of possible \emph{assignments};
The \emph{domain} of $v \in X$ is its non-empty set $\mathcal{D}(v)$
of possible  values: $\mathcal{D}(v) = \{a \in D \mid \tuple{v,a} \in
A\}$; and
\item $\cpt$ is a partial \emph{compatibility function} from the set
of unordered pairs of assignments $\{\{\tuple{v,a},\tuple{w,b}\} \mid
v \neq w\}$ to $\{{\tt TRUE},{\tt FALSE}\}$; if
$\cpt(\tuple{v,a},\tuple{w,b})$ $=$ {\tt TRUE} (resp., {\tt FALSE})
we say that \tuple{v,a} and \tuple{w,b} are \emph{compatible} (resp.,
\emph{incompatible}).
\end{itemize}

A \emph{quantified pattern} is a pattern $P$ with a distinguished
variable, $\overline{v}(P)$ and a subset of existential values $e(P)
\subseteq \mathcal{D}(\overline{v}(P))$.

A \emph{flat quantified pattern} is a quantified pattern for which
$e(P)$ is empty. An \emph{existential pattern} is a quantified
pattern $P$ for which $e(P)$ is non-empty. An existential pattern $P$
may also have a distinguished value $\overline{val}(P) \in e(P)$.
\end{definition}

When the context variable $v$ is clear we use the value $a$ to denote
the assignment \tuple{v,a} to $v$.  We will often simplify notation
by writing $\cpt(p,q)$ for $\cpt(\{p,q\})$. We will also use the
terminology of graph theory, since a pattern can be viewed as a
labelled graph: if $\cpt(p,q) =$ {\tt TRUE} (resp., {\tt FALSE}),
then we say that there is a \emph{compatibility} (resp.,
\emph{incompatibility}) \emph{edge} between $p$ and $q$.

We will use a simple figurative drawing for patterns.  Each variable
will be drawn as an oval containing dots for each of its possible
assignments. Pairs in the domain of the function \cpt\ will be
represented by lines between values: solid lines for compatibility
and dashed lines for incompatibility.  The distinguished variable
($\overline{v}(P)$) and any existential values in $e(P)$ will be
indicated by an $\exists$ symbol. Examples of patterns are shown in
Figure~\ref{fig:Examplepatterns} and Figure~\ref{fig:VEpatterns}.

We are never interested in the names of variables nor the names of
the domain values in patterns. So we define the following
equivalence.
\begin{definition}
Two patterns $P$ and $Q$ are \emph{equivalent} if they are
isomorphic, {\it i.e.} if they are identical except for possible
injective renamings of variables and assignments which preserve
$\mathcal{D}$, \cpt, $\overline{v}$, $e$ and $\overline{val}$.
\end{definition}

A pattern can be viewed as a CSP instance in which not all
compatibilities are defined. We can thus refine patterns to give a
definition of a (binary) CSP instance.

\begin{definition}
A \emph{binary CSP instance} $P$ is a pattern \tuple{X, D, A, \cpt}
where \cpt\ is a total function, {\it i.e.} the domain of \cpt\ is
precisely $\{\{\tuple{v,a},\tuple{w,b}\} \mid v \neq w$, $a \in
\mathcal{D}(v)$, $b \in \mathcal{D}(w)\}$.
\begin{itemize}
\item The \emph{relation} $R_{v,w} \subseteq \mathcal{D}(v) \times \mathcal{D}(w)$
on $\tuple{v,w}$ is $\{\tuple{a,b}\mid \cpt(\tuple{v,a},\tuple{w,b})
= {\tt TRUE}\}$.

\item A \emph{partial solution} to $P$ on $Y \subseteq X$ is a
mapping $s: Y \to D$ where, for all $v \neq w \in Y$ we have
$\tuple{s(v), s(w)} \in R_{v,w}$.

\item A \emph{solution} to $P$ is a partial solution on $X$.
\end{itemize}
\end{definition}

For notational simplicity we have assumed that there is exactly one binary
constraint between each pair of variables. In particular, this means
that the absence of a constraint between variables $v,w$ is modelled
by a complete relation $R_{v,w} = \mathcal{D}(v) \times
\mathcal{D}(w)$ allowing every possible pair of assignments to $v$
and $w$. We say that there is a \emph{non-trivial} constraint on
variables $v,w$ if $R_{v,w} \neq \mathcal{D}(v) \times
\mathcal{D}(w)$.

In practice, when solving CSP instances we prune the domains of
variables in such a way as to maintain all solutions.
\begin{definition}
Let $P = \tuple{X,D,A,\cpt}$ be a CSP instance. An assignment
$\tuple{v,a} \in A$ to variable $v$ is called \emph{arc consistent}
if, for all  variables $w \neq v$ in $X$ there is some assignment
$\tuple{w,b} \in A$ compatible  with \tuple{v,a}.

The CSP instance \tuple{X,D,A,\cpt} is called \emph{arc consistent}
if every assignment in $A$ is arc consistent.
\end{definition}

Assignments that are not arc-consistent cannot be part of a solution
so can safely be removed. There are optimal $O(cd^2)$ algorithms for
establishing arc consistency which repeatedly remove such
values~\cite{Bessiere:AC}, where $c$ is the number of non-trivial
constraints and $d$ the maximum domain size. Hence, for the remainder
of this paper we will assume that all CSP instances are
arc-consistent.

In this paper we are concerned with variable elimination
characterised by forbidden patterns. We now define what this means.

\begin{definition} \label{def:varelim}
We say that a variable $x$ can be \emph{eliminated} in the CSP
instance \tuple{X, D, A, \cpt} if, whenever there is a partial
solution on $X \setminus \{x\}$ there is a solution.
\end{definition}

In order to use (the absence of) patterns for variable elimination we
need to define what we mean when we say that a quantified pattern
occurs at variable $x$ of a CSP instance.  We define occurrence in
terms of reductions on patterns. The definitions of occurrence and
reduction between quantified patterns extend definitions previously
given for non-quantified patterns~\cite{Cooper12:patterns}.

\begin{definition} \label{def:reductions}
Let $P =\tuple{X,D,A,\cpt}$ be any pattern.
\begin{itemize}
\item We say that a pattern $P' =\tuple{X',D',A',\cpt'}$ is a
\emph{sub-pattern} of $P$ if $X' \subseteq X, A' \subseteq A$ and
$\forall p,q \in A'$, either $\cpt'(p,q) = \cpt(p,q)$ or $\cpt'(p,q)$ is undefined.

If, furthermore, $P'$ is quantified then we require that $P$ is
quantified and that $\overline{v}(P') = \overline{v}(P)$ and $e(P')
\subseteq e(P)$. If $P'$ has a distinguished value then we require
that $P$ also has a distinguished value and that $\overline{val}(P')
= \overline{val}(P)$.

\item Values $a,b \in \mathcal{D}(v)$ are \emph{mergeable} in a pattern
if there is no assignment $p \in A$ for which $\cpt(\tuple{v,a},p)$,
$\cpt(\tuple{v,b},p)$ are both defined and $\cpt(\tuple{v,a},p) \neq
\cpt(\tuple{v,b},p)$.
In a quantified pattern, for $a$ to be merged into $b$, we also
require that $a \in e(P)$ only if $b \in e(P)$.

When $a,b \in \mathcal{D}(v)$ are mergeable we define the merge
reduction $\tuple{X, D,  A\setminus\{\tuple{v,a}\},\cpt'}$, in which
$a$ is merged into $b$, by the following compatibility function:
\begin{equation*}
\cpt'(p,q) =
\begin{cases}
\cpt(\tuple{v,a},q) & \text{if  } p = \tuple{v,b} \text{  and  } 
\text{$\cpt(p,q)$ undefined,}\\
\cpt(p,q) & \text{otherwise.}
\end{cases}
\end{equation*}

\item A \emph{dangling assignment} $p$ of $P$ is any assignment for which
there is at most one assignment $q$ for which $\cpt(p,q)$ is defined,
and furthermore (if defined) $\cpt(p,q)= {\tt TRUE}$.
If $P$ is quantified, then we also require that
$p \notin \overline{v}(P) \times e(P)$. For any
dangling assignment $p$, we define the dangling reduction
$\tuple{X,D,A',\cpt\restriction_{A' \times A'}}$ where $A'=A \setminus \{p\}$.

\item A \emph{reduction} of a pattern
$P$ is a pattern obtained from $P$ by a sequence of merge and
dangling reductions. An \emph{irreducible pattern} is one on which no
merge or dangling reductions can be performed.
\end{itemize}
\end{definition}

To illustrate the notions introduced in Definition~\ref{def:reductions}, 
consider the patterns in Figure~\ref{fig:Examplepatterns}.
Pattern $P_1$ is a sub-pattern of $P_2$ which is itself a sub-pattern of $P_3$.
In pattern $P_2$, the values $a,b \in \mathcal{D}(x)$ are mergeable: merging
$a$ into $b$ produces the pattern $P_4$. In the pattern $P_3$, the values 
 $a,b \in \mathcal{D}(x)$ are not mergeable since 
$\cpt(\tuple{x,a},\tuple{z,d})$ and $\cpt(\tuple{x,b},\tuple{z,d})$ are both
defined but are not equal.
In pattern $P_2$, $\tuple{x,a}$ is a dangling assignment: applying the dangling reduction to
this assignment in $P_2$ produces $P_1$. Let $P_2'$ be identical to $P_2$ except that
$P_2'$ is a quantified pattern with $\overline{v}(P_2') = \{x\}$ and $e(P_2') = \{a\}$.
Then $P_2$ is a sub-pattern of $P_2'$, but $P_2'$ is not a sub-pattern of $P_2$.
In the domain of $x$ in $P_2'$, $b$ can be merged into $a$ but $a$ cannot be merged
into $b$ since $a \in e(P_2')$ but $b \notin e(P_2')$. Furthermore, the assignment
$\tuple{x,a}$ is not a dangling assignment in $P_2'$ since $a$ is an existential value
for $\overline{v}(P_2') = x$.

\setlength{\unitlength}{1pt}
\begin{figure}
\centering

\begin{picture}(300,240)(0,0)

\put(0,130){
\begin{picture}(140,120)(0,-10)
\put(10,40){\usebox{\varone}} \put(50,0){\usebox{\varone}}
\put(90,40){\usebox{\varone}}
\dashline{5}(20,60)(100,60) \dashline{5}(20,60)(60,20)
\put(100,40){\makebox(0,0){$x$}}
 \put(20,40){\makebox(0,0){$y$}}  \put(60,0){\makebox(0,0){$z$}}
\put(60,80){\makebox(0,0){$P_1$}}   \put(115,60){\makebox(0,0){$b$}}
  \put(5,60){\makebox(0,0){$c$}}
 \put(45,20){\makebox(0,0){$d$}}
\end{picture}}

\put(150,130){
\begin{picture}(140,120)(0,-10)
\put(10,40){\usebox{\varone}} \put(50,0){\usebox{\varone}}
\put(90,40){\usebox{\vartwo}} \put(60,20){\line(4,3){40}}
\dashline{5}(20,60)(60,20)
\dashline{5}(20,60)(100,70) \put(100,35){\makebox(0,0){$x$}}
 \put(20,40){\makebox(0,0){$y$}}  \put(60,0){\makebox(0,0){$z$}}
\put(60,80){\makebox(0,0){$P_2$}}   \put(115,50){\makebox(0,0){$a$}}
 \put(115,70){\makebox(0,0){$b$}}  \put(5,60){\makebox(0,0){$c$}}
 \put(45,20){\makebox(0,0){$d$}}
\end{picture}}

\put(0,0){
\begin{picture}(140,120)(0,-10)
\put(10,40){\usebox{\varone}} \put(50,0){\usebox{\varone}}
\put(90,40){\usebox{\vartwo}} \put(60,20){\line(4,3){40}}
\dashline{5}(20,60)(60,20) \dashline{5}(100,70)(60,20)
\dashline{5}(20,60)(100,70) \put(100,35){\makebox(0,0){$x$}}
 \put(20,40){\makebox(0,0){$y$}}  \put(60,0){\makebox(0,0){$z$}}
\put(60,80){\makebox(0,0){$P_3$}}   \put(115,50){\makebox(0,0){$a$}}
 \put(115,70){\makebox(0,0){$b$}}  \put(5,60){\makebox(0,0){$c$}}
 \put(45,20){\makebox(0,0){$d$}}
\end{picture}}

\put(150,0){
\begin{picture}(140,120)(0,-10)
 \put(60,20){\line(1,1){40}}
\put(10,40){\usebox{\varone}} \put(50,0){\usebox{\varone}}
\put(90,40){\usebox{\varone}}
\dashline{5}(20,60)(100,60) \dashline{5}(20,60)(60,20)
\put(100,40){\makebox(0,0){$x$}}
 \put(20,40){\makebox(0,0){$y$}}  \put(60,0){\makebox(0,0){$z$}}
\put(60,80){\makebox(0,0){$P_4$}}   \put(115,60){\makebox(0,0){$b$}}
  \put(5,60){\makebox(0,0){$c$}}
 \put(45,20){\makebox(0,0){$d$}}
\end{picture}}

\end{picture}

\caption{Examples illustrating the notions of sub-pattern, merging and dangling assignment.}

\label{fig:Examplepatterns}

\end{figure}
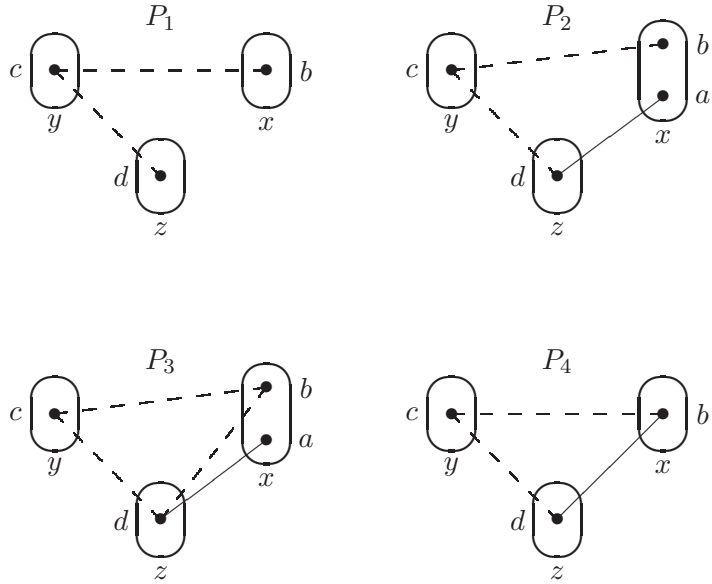
\setlength{\unitlength}{1pt}

Now we want to define when a quantified pattern occurs at a variable
in a CSP instance, in order to characterise those patterns whose
\emph{non-occurrence} allows this particular variable to be
eliminated. We define the slightly more general notion of occurrence
of a pattern in another pattern. Recall that a CSP instance
corresponds to the special case of a pattern whose compatibility
function is total. Essentially we want to say that pattern $P$ occurs
in pattern $Q$ if $P$ is homomorphic to a sub-pattern of $Q$ via an
injective renaming of variables and a (possibly non-injective)
renaming of assignments~\cite{Cohen12:pivot}. However, we find it
simpler to define occurrence using the notions of sub-pattern,
reduction and equivalence. We first make the observation that
dangling assignments in a pattern provide no useful information since
we assume that all CSP instances are arc consistent, which explains
why dangling assignments can be eliminated from patterns.

We can then define occurrence in terms of reduced patterns.

\begin{definition}\label{def:occurs}
We say that a pattern $P$ \emph{occurs} in a pattern $Q$
(and that $Q$ \emph{contains} $P$) if some reduction of $P$ is
equivalent to a sub-pattern of $Q$.

If $Q$ is a CSP instance, then the quantified pattern $P$ 
\emph{occurs at variable $x$ of
$Q$} if some reduction of $P$ is equivalent to a sub-pattern of $Q$
and $x$ is the variable of the sub-pattern of $Q$ corresponding to
$\overline{v}(P)$.

We say that the quantified pattern $P$ 
\emph{occurs at variable $x$ of $Q$ with value
mapping} $m:e(P) \rightarrow \mathcal{D}(x)$ if the values of
variable $x$ corresponding to each $a \in e(P)$ are given by the
mapping $m$.
\end{definition}

A variable elimination pattern is defined in terms of occurrence of a
pattern in a CSP instance.

\begin{definition} \label{def:VE}
A quantified pattern is a \emph{variable elimination pattern
(var-elim pattern)} if, whenever the pattern does not occur at a
variable $x$ in an arc-consistent CSP instance $I$ for at least one
injective value mapping, $x$ can be eliminated in $I$ (in the sense
of Definition~\ref{def:varelim}).

A non-quantified pattern (i.e. a pattern without a distinguished
variable) is a var-elim pattern if, whenever the pattern does not
occur in an arc-consistent CSP instance, \emph{any} variable can be
eliminated in $I$.
\end{definition}

The notion of non-quantified var-elim patterns is necessary for some
of our proofs, but for practical applications we are interested in
finding quantified (and, in particular, existential) var-elim
patterns. Existential patterns may allow more variables to be
eliminated than flat quantified patterns. For example, as we will
show later, the patterns snake and $\exists$snake shown in
Figure~\ref{fig:VEpatterns} are both var-elim patterns, but the
latter allows more variables to be eliminated since we only require
that it does not occur on a single value in the domain of the
variable to be eliminated.

\begin{example}
Suppose that we can assign value 0 to a subset $S$ of the
variables of an instance, without restricting the assignments to any other
variables. Furthermore suppose that, within $S$, 0 is only compatible
with 0. The var-elim pattern $\exists$invsubBTP, shown in
Figure~\ref{fig:VEpatterns}, allows us to eliminate all variables in
$S$, without having to explicitly search for $S$. This is because the
pattern does not occur for the mapping $a \mapsto 0$. The flat
variant (invsubBTP) would not allow these eliminations.
\end{example}

We conclude this section with the simple observation that var-elim
patterns define tractable classes.  It takes polynomial time to
establish arc consistency and to detect (by exhaustive search) the
non-occurrence of a var-elim pattern. Hence it takes polynomial time
to identify arc-consistent CSP instances for which all variables can
be eliminated one by one by a var-elim pattern $P$. Such instances
are solvable in a greedy fashion.

Hence we are able to significantly extend the list of known tractable
classes defined by forbidden patterns since among known tractable
patterns, namely BTP~\cite{Cooper10:BTP}, 2-constraint
patterns~\cite{Cooper12:patterns}, pivots~\cite{Cohen12:pivot} and
JWP~\cite{Cooper12:tractable}, \textit{only} BTP (and its sub-patterns) allow
variable elimination.

Indeed, a general hybrid tractable class can be defined: the set of
binary CSP instances which fall in some known tractable class after
we have performed all variable (and value) eliminations defined by
the rules given in this paper.

\section{Variable elimination by forbidden patterns}
\label{sec:VEpatterns}

In this paper we characterise irreducible var-elim patterns. There
are essentially just four (together with their irreducible
sub-patterns): the patterns BTP, $\exists$subBTP,
$\exists$invsubBTP and $\exists$snake, shown in
Figure~\ref{fig:VEpatterns}. We begin by showing that each of these four patterns
allows variable elimination. Forbidding BTP is equivalent to the
already-known local Broken Triangle Property
(lBTP)~\cite{Cooper10:BTP} mentioned in Section~\ref{subsec:vvelim}.

\setlength{\unitlength}{1pt}
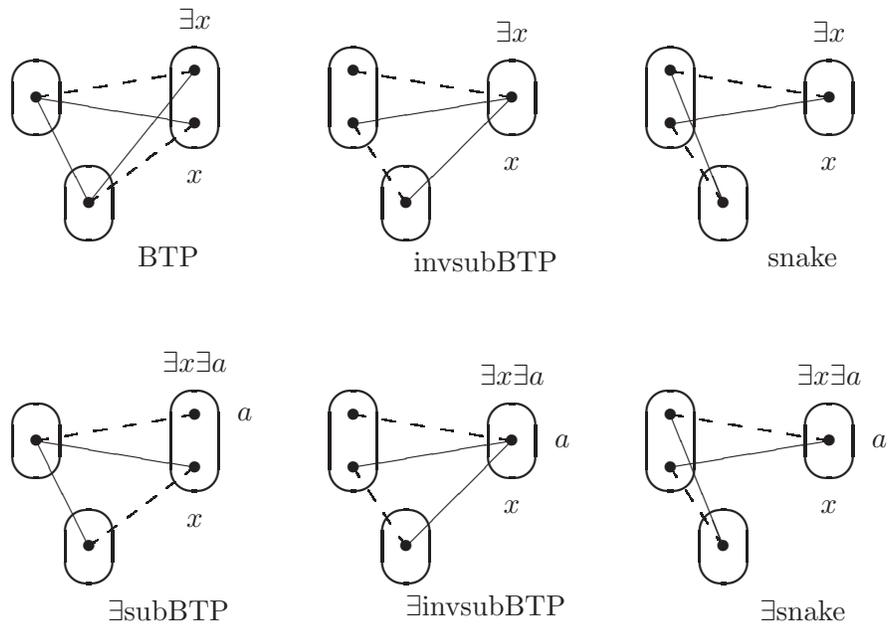
\begin{figure}
\centering

\begin{picture}(338,240)(2,0)

\put(0,130){
\begin{picture}(100,110)(0,0)
\put(0,50){\usebox{\varone}} \put(20,10){\usebox{\varone}}
\put(60,50){\usebox{\vartwo}} \put(10,70){\line(6,-1){60}}
\put(10,70){\line(1,-2){20}} \dashline{5}(30,30)(70,60)
\dashline{5}(10,70)(70,80) \put(70,40){\makebox(0,0){$x$}}
\put(70,100){\makebox(0,0){$\exists x$}}
\put(60,10){\makebox(0,0){BTP}} \put(30,30){\line(4,5){40}}
\end{picture}}

\put(120,130){
\begin{picture}(100,110)(0,0)
\put(0,50){\usebox{\vartwo}} \put(20,10){\usebox{\varone}}
\put(60,50){\usebox{\varone}} \put(10,60){\line(6,1){60}}
\dashline{5}(30,30)(10,60) \dashline{5}(10,80)(70,70)
\put(70,45){\makebox(0,0){$x$}} \put(70,95){\makebox(0,0){$\exists
x$}} \put(60,8){\makebox(0,0){invsubBTP}} \put(30,30){\line(1,1){40}}
\end{picture}}

\put(240,130){
\begin{picture}(100,110)(0,0)
\put(0,50){\usebox{\vartwo}} \put(20,10){\usebox{\varone}}
\put(60,50){\usebox{\varone}} \put(10,60){\line(6,1){60}}
\put(10,80){\line(2,-5){20}} \dashline{5}(30,30)(10,60)
\dashline{5}(10,80)(70,70) \put(70,45){\makebox(0,0){$x$}}
\put(70,95){\makebox(0,0){$\exists x$}}
\put(60,10){\makebox(0,0){snake}}
\end{picture}}

\put(0,0){
\begin{picture}(100,110)(0,0)
\put(0,50){\usebox{\varone}} \put(20,10){\usebox{\varone}}
\put(60,50){\usebox{\vartwo}} \put(10,70){\line(6,-1){60}}
\put(10,70){\line(1,-2){20}} \dashline{5}(30,30)(70,60)
\dashline{5}(10,70)(70,80) \put(70,40){\makebox(0,0){$x$}}
\put(89,80){\makebox(0,0){$a$}} \put(70,100){\makebox(0,0){$\exists x
\exists a$}} \put(60,6){\makebox(0,0){$\exists$subBTP}}
\end{picture}}

\put(120,0){
\begin{picture}(100,110)(0,0)
\put(0,50){\usebox{\vartwo}} \put(20,10){\usebox{\varone}}
\put(60,50){\usebox{\varone}} \put(10,60){\line(6,1){60}}
\dashline{5}(30,30)(10,60) \dashline{5}(10,80)(70,70)
\put(70,45){\makebox(0,0){$x$}} \put(89,70){\makebox(0,0){$a$}}
\put(70,95){\makebox(0,0){$\exists x \exists a$}}
\put(60,10){\makebox(0,-4){$\exists$invsubBTP}}
\put(30,30){\line(1,1){40}}
\end{picture}}

\put(240,0){
\begin{picture}(100,110)(0,0)
\put(0,50){\usebox{\vartwo}} \put(20,10){\usebox{\varone}}
\put(60,50){\usebox{\varone}} \put(10,60){\line(6,1){60}}
\put(10,80){\line(2,-5){20}} \dashline{5}(30,30)(10,60)
\dashline{5}(10,80)(70,70) \put(70,45){\makebox(0,0){$x$}}
\put(89,70){\makebox(0,0){$a$}} \put(70,95){\makebox(0,0){$\exists x
\exists a$}} \put(60,6){\makebox(0,0){$\exists$snake}}
\end{picture}}

\end{picture}

\caption{Variable elimination patterns.}

\label{fig:VEpatterns}

\end{figure}
\setlength{\unitlength}{1pt}

\begin{theorem} \label{thm:allVEpatterns}
The patterns BTP, $\exists$subBTP, $\exists$invsubBTP and
$\exists$snake are var-elim patterns.
\end{theorem}

\begin{proof}
Since it is known that BTP is a var-elim
pattern~\cite{Cooper10:BTP}, we only need to prove the result for the
three existential patterns: $\exists$subBTP, $\exists$invsubBTP and
$\exists$snake.

Every two-variable arc-consistent CSP instance allows either variable
to be eliminated. So we only have to prove that these patterns allow
variable elimination in CSP instances with at least three variables.

We first set up some general machinery which will be used in each of the three cases.
Consider an arc-consistent CSP instance $I=\tuple{X,D,A,\cpt}$ and
let $s$ be a partial solution on $X \setminus \{x\}$.

Fix some assignment \tuple{x,d}, and let:
\begin{eqnarray*}
Y  &=  \left\{y \in X \setminus \{x\} \mid \cpt(\tuple{y,s(y)},\tuple{x,d}) =
 {\tt TRUE}\right\},\\
\overline{Y}  &=  \left\{z \in X \setminus \{x\} \mid
\cpt(\tuple{z,s(z)},\tuple{x,d}) = {\tt FALSE}\right\}.
\end{eqnarray*}

For all $y,z \in X \setminus \{x\}$, since $s$ is a partial solution, 
$\cpt(\tuple{y,s(y)},\tuple{z,s(z)}) =$ {\tt TRUE}. Thus, if $X = Y
\cup \{x\}$ then we can extend $s$ to a solution to $I$ by choosing
value $d$ for variable $x$.  So, in this case $x$ could be
eliminated. So we assume from now on that $\overline{Y} \neq
\emptyset$.

By arc consistency, for all $z \in \overline{Y}$, there is some
$\tuple{z,t(z)} \in A$ such that $\cpt(\tuple{z,t(z)},\tuple{x,d}) =$
{\tt TRUE}.

We now prove the result for each pattern in turn.

Suppose that $\exists$subBTP does not occur at $x$ in $I$ for the
mapping $a \mapsto d$. Consider any $y \in \overline{Y}$. By arc
consistency, $\exists b \in \mathcal{D}(x)$ such that
$\cpt(\tuple{y,s(y)},\tuple{x,b}) =$ {\tt TRUE}. Since the pattern
$\exists$subBTP does not occur, and in particular on the set of
assignments $\{\tuple{y,s(y)}$, $\tuple{z,s(z)}$, $\tuple{x,d}$,
$\tuple{x,b}\}$, we can deduce that, for every variable $z \in X$
different from both $x$ and $y$, $\cpt(\tuple{z,s(z)},\tuple{x,b}) =$
{\tt TRUE}. Hence, we can extend $s$ to a solution to $I$ by choosing
$s(x) =b$.  So, in any case $x$ can be eliminated and $\exists$subBTP
is indeed a var-elim pattern.

Now instead, suppose $\exists$invsubBTP does not occur at $x$ in $I$
for the mapping $a \mapsto d$. Since the pattern $\exists$invsubBTP
does not occur, if both $y$ and $z$ belong to $\overline{Y}$ then
$\cpt(\tuple{y,t(y)},\tuple{z,t(z)}) =$ {\tt TRUE}, otherwise the
pattern would occur on the assignments $\{\tuple{y,s(y)}$,
$\tuple{y,t(y)}$, $\tuple{z,t(z)}$, $\tuple{x,d}\}$. Also, if $y \in
Y$, $z \in \overline{Y}$, then $\cpt(\tuple{y,s(y)},\tuple{z,t(z)})
=$ {\tt TRUE}, otherwise the pattern would occur on
$\{\tuple{z,s(z)}$, $\tuple{z,t(z)}$, $\tuple{y,s(y)}$,
$\tuple{x,d}\}$.

So, in this case we have a solution $s'$ to $I$, where
\begin{equation*}
s'(v)=
\begin{cases}
d & \text{if $v = x$,}\\
s(v) & \text{if $v \in Y$,}\\
t(v) & \text{otherwise.}
\end{cases}
\end{equation*}
So $\exists$invsubBTP is indeed a var-elim pattern.

For the final pattern, suppose that $\exists$snake does not occur at
$x$ in $I$ for the mapping $a \mapsto d$. If $y \in Y$, $z \in
\overline{Y}$, since the pattern $\exists$snake does not occur, we
can deduce that $\cpt(\tuple{y,s(y)},\tuple{z,t(z)}) =$ {\tt TRUE},
otherwise the pattern would occur on the assignments
$\{\tuple{z,s(z)}$, $\tuple{z,t(z)}$, $\tuple{y,s(y)}$,
$\tuple{x,d}\}$. If both $y$ and $z$ both belong to $\overline{Y}$,
then we can deduce first that $cpt(\tuple{y,s(y)},\tuple{z,t(z)}) =$
{\tt TRUE} (as in the previous case) and then, as a consequence, that
$cpt(\tuple{y,t(y)},\tuple{z,t(z)}) =$ {\tt TRUE} (otherwise the
pattern would occur on $\{\tuple{y,s(y)}$, $\tuple{y,t(y)}$,
$\tuple{z,t(z)}$, $\tuple{x,d}\}$).

So, again in this case we have a solution $s'$ to $I$, where $s'$ is
defined as above. So $\exists$snake is also a var-elim pattern.
\end{proof}

\section{Characterisation of quantified var-elim patterns}

Our aim is to precisely characterise all irreducible patterns which
allow variable elimination in an arc-consistent binary CSP instance.
We begin by identifying many patterns, including all those shown in
Figure~\ref{fig:nonVEpatterns}, which are not variable elimination
patterns.

\setlength{\unitlength}{1pt}
\begin{figure}
\centering

\begin{picture}(338,419)(0,-5)

\put(0,0){
\begin{picture}(100,100)(0,-10)
\put(0,40){\usebox{\varone}} \put(20,0){\usebox{\varone}}
\put(60,40){\usebox{\varone}} \dashline{5}(30,20)(70,60)
\dashline{5}(10,60)(70,60) \put(70,37){\makebox(0,0){$x$}}
\put(70,83){\makebox(0,0){$\exists x$}}
\put(60,-5){\makebox(0,0){Pivot(sym)}}
\end{picture}}

\put(120,0){
\begin{picture}(100,100)(0,-10)
\put(0,40){\usebox{\varone}} \put(20,0){\usebox{\varone}}
\put(60,40){\usebox{\varone}} \dashline{5}(30,20)(10,60)
\dashline{5}(10,60)(70,60) \put(70,37){\makebox(0,0){$x$}}
\put(70,83){\makebox(0,0){$\exists x$}}
\put(60,-5){\makebox(0,0){Pivot(asym)}}
\end{picture}}

\put(240,0){
\begin{picture}(100,110)(0,0)
\put(0,50){\usebox{\vartwo}} \put(60,50){\usebox{\vartwo}}
\put(30,10){\usebox{\vartwo}} \dashline{5}(10,80)(70,80)
\dashline{5}(10,60)(40,40) \dashline{5}(70,60)(40,20)
\put(70,3){\makebox(0,0){Cycle(3)}}
\end{picture}}

\put(0,120){
\begin{picture}(100,110)(0,0)
\put(0,50){\usebox{\vartwo}} \put(60,50){\usebox{\varone}}
\put(30,10){\usebox{\vartwo}} \put(10,60){\line(6,1){60}}
\put(10,60){\line(3,-2){30}} \put(10,80){\line(1,-2){30}}
\put(40,20){\line(3,5){30}} \dashline{5}(10,80)(40,40)
\put(70,47){\makebox(0,0){$x$}} \put(70,93){\makebox(0,0){$\exists
x$}} \put(80,8){\makebox(0,0){Kite(sym)}}
\end{picture}}

\put(120,120){
\begin{picture}(100,110)(0,0)
\put(0,50){\usebox{\vartwo}} \put(60,50){\usebox{\vartwo}}
\put(30,10){\usebox{\varone}} \put(10,60){\line(3,1){60}}
\put(10,60){\line(1,-1){30}} \put(10,80){\line(3,-1){60}}
\put(40,30){\line(1,1){30}} \dashline{5}(10,80)(70,80)
\put(70,42){\makebox(0,0){$x$}} \put(70,98){\makebox(0,0){$\exists
x$}} \put(77,8){\makebox(0,0){Kite(asym)}}
\end{picture}}

\put(240,120){
\begin{picture}(100,110)(0,0)
\put(0,50){\usebox{\vartwo}} \put(20,10){\usebox{\varone}}
\put(60,50){\usebox{\varone}} \put(10,80){\line(2,-5){20}}
\dashline{5}(30,30)(10,60) \dashline{5}(10,80)(70,70)
\put(70,47){\makebox(0,0){$x$}} \put(70,93){\makebox(0,0){$\exists
x$}} \put(60,10){\makebox(0,0){rotsubBTP}}
\put(30,30){\line(1,1){40}}
\end{picture}}

\put(0,240){
\begin{picture}(100,115)(0,-20)
\put(0,10){\usebox{\varone}} \put(60,10){\usebox{\vartwo}}
\put(10,30){\line(6,1){60}} \dashline{5}(10,30)(70,20)
\put(70,2){\makebox(0,0){$x$}} \put(88,40){\makebox(0,0){$a$}}
\put(70,58){\makebox(0,0){$\exists x \exists a$}}
\put(40,-15){\makebox(0,0){V($+ -$)}}
\end{picture}}

\put(100,240){
\begin{picture}(110,115)(-20,0)
\put(0,50){\usebox{\varone}} \put(60,50){\usebox{\varone}}
\put(30,10){\usebox{\varone}} \dashline{5}(10,70)(70,70)
\put(10,70){\line(3,-4){30}} \put(40,30){\line(3,4){30}}
\put(70,47){\makebox(0,0){$x$}} \put(88,70){\makebox(0,0){$a$}}
\put(70,93){\makebox(0,0){$\exists x \exists a$}}
\put(45,5){\makebox(0,0){Triangle(asym)}}
\end{picture}}

\put(240,240){\begin{picture}(100,100)(0,0)
\put(0,50){\usebox{\varone}} \put(60,50){\usebox{\varone}}
\put(30,10){\usebox{\varone}} \put(10,70){\line(1,0){60}}
\put(10,70){\line(3,-4){30}} \put(40,30){\line(3,4){30}}
\put(70,5){\makebox(0,0){Triangle}}
\end{picture}}

\put(0,340){
\begin{picture}(100,80)(0,-15)
\put(0,20){\usebox{\varone}} \put(30,20){\usebox{\vartwo}}
\put(60,20){\usebox{\varone}} \put(40,30){\line(3,1){30}}
\put(40,50){\line(3,-1){30}} \put(40,30){\line(-3,1){30}}
\dashline{5}(40,50)(10,40) \put(40,5){\makebox(0,0){Diamond}}
\end{picture}}

\put(120,340){
\begin{picture}(110,80)(0,20)
\put(0,50){\usebox{\vartwo}} \put(60,50){\usebox{\vartwo}}
\dashline{5}(10,80)(70,60) \put(10,80){\line(1,0){60}}
\put(10,60){\line(3,1){60}} \put(10,60){\line(1,0){60}}
\put(40,40){\makebox(0,0){Z}}
\end{picture}}

\put(240,340){\begin{picture}(100,80)(0,15)
\put(0,50){\usebox{\vartwo}} \put(30,10){\usebox{\varone}}
\put(60,50){\usebox{\vartwo}} \dashline{5}(10,80)(70,80)
\dashline{5}(40,30)(10,60) \put(10,80){\line(3,-1){60}}
\put(10,60){\line(3,1){60}} \put(40,30){\line(1,1){30}}
\put(65,25){\makebox(0,0){XL}}
\end{picture}}

\end{picture}

\caption{Patterns which do not allow variable elimination.}

\label{fig:nonVEpatterns}

\end{figure}
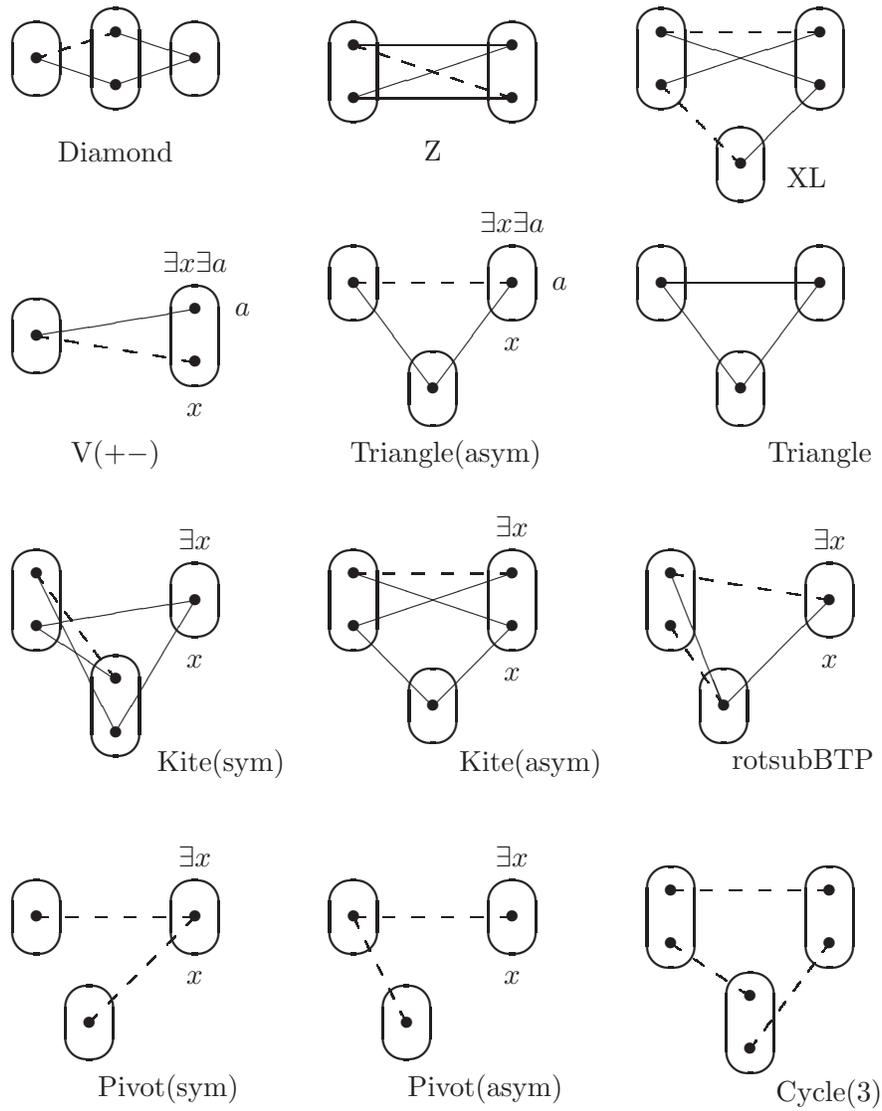
\setlength{\unitlength}{1pt}

\begin{lemma} \label{lem:nonVEpatterns}
None of the following patterns allow variable elimination in
arc-consistent binary CSP instances: any pattern on strictly more
than three variables, any pattern with three non-mergeable values for
the same variable, any pattern with two non-mergeable incompatibility
edges in the same constraint, {\nobreak Diamond}, Z, XL, V($+ -$),
Triangle(asym), Triangle, Kite(sym), Kite(asym), \ {\nobreak
rotsubBTP}, Pivot(asym), Pivot(sym), Cycle(3).
\end{lemma}

\begin{proof}
For each pattern we exhibit a binary arc-consistent CSP instance
that:
\begin{itemize}
\item has a partial solution on the set of all the variables except a specified variable $x$;
\item has no solution;
\item does not contain the given pattern $P$
at variable $x$ (if $P$ is a quantified pattern) or does not contain
$P$ at any variable (if $P$ is a non-quantified pattern).
\end{itemize}

By definition, any such instance is enough to prove that a pattern is
not a var-elim pattern.

\begin{itemize}
\item For any pattern $P$ which is either Diamond, Z, XL,
or Triangle, or has at least four variables, or has three
non-mergeable values for the same variable.

Let $I^{2COL}_3$ be the CSP instance (corresponding to 2-colouring on
3 variables) with three Boolean variables, where the constraint
between any two variables forces them to take different values.

This instance has partial solutions on any two variables, but has no
solution, and does not contain $P$.

\item For V($+ -$) and Triangle(asym).

Let $I^{\exists}_4$ be the instance on four variables $x_1,x_2,x_3$
and $x$, where the domains of $x_1,x_2$ and $x_3$ are all $\{0,1,2\}$
and the domain of $x$ is $\{0,1,2,3\}$. Each pair of variables in
$\{x_1,x_2,x_3\}$ must take values in $\{\tuple{0,0}, \tuple{1,2},
\tuple{2,1}\}$.  There are three further constraints: for $i=1,2,3$,
we have that $(x_i>0) \vee (x=i)$.

$I^{\exists}_4$ has a partial solution on $\{x_1,x_2,x_3\}$ but has
no solution. $I^{\exists}_4$ contains neither V($+ -$) nor
Triangle(asym) at variable $x$ for the value mapping $m(a) = 0$.

\item
For Kite(sym).

Let $I_4$ be the CSP instance on four variables $x_1,x_2,x_3,x$ where
$x_1, x_2$ and $x_3$ are Boolean and $\mathcal{D}(x) = \{1,2,3\}$,
with the following constraints: $x_1 \vee x_2$, $x_1 \vee x_3$, $x_2
\vee x_3$, $x_i \Leftrightarrow (x = i)$ ($i=1,2,3$).

$I_4$ has a partial solution on $\{x_1,x_2,x_3\}$, has no solution,
and does not contain Kite(sym) at variable $x$.

\item
For Kite(asym).

Let $I^{ZOA}_4$ be the CSP instance on the four variables
$x_1,x_2,x_3,x$ each with domain $\{1,2,3\}$, with the following
constraints: $x_1=x_2$, $x_1=x_3$, $x_2=x_3$, $(x_1=1) \vee (x=1)$,
$(x_2=2) \vee (x=2)$, $(x_3=3) \vee (x=3)$.

\item
For rotsubBTP.

Define the three binary relations:
\begin{eqnarray*}
R &=
\{\tuple{0,0}, \tuple{1,2}, \tuple{2,1}\},\\
R_0 &= \{\tuple{0,0}, \tuple{1,1}, \tuple{2,1}\},\\
R_1 &= \{\tuple{0,1}, \tuple{1,0}, \tuple{2,0}\}.
\end{eqnarray*}

Let $I_7$ be the CSP instance on the seven variables
$x_1,\ldots,x_6,x$ where $\mathcal{D}(x_i) = \{0,1,2\}$, for
$i=1,\dots,6$, and $\mathcal{D}(x) = \{0,1\}$, with the following
constraints:

For $(1 \leq i < j \leq 3)$ and $(4 \leq i < j \leq 6)$,
\tuple{x_i,x_j} must take values in $R$.

For $(1 \leq i \leq 3)$, \tuple{x_i,x} must take values in $R_0$.

For $(4 \leq i \leq 6)$, \tuple{x_i,x} must take values in $R_1$ .

\item For the pattern Pivot(sym).

Let $I^{SAT}_4$ be the 2SAT instance on four Boolean variables
$x_1,x_2,x_3,x$ with the following constraints: $x_1 \equiv x_2$,
$x_1 \equiv x_3$, $x_2 \vee x_3$, $\overline{x_2} \vee x$,
$\overline{x_3} \vee \overline{x}$.

\item
For
Cycle(3) or Pivot(asym), or any pattern with two non-mergeable
incompatibility edges in the same constraint.

Let $I^{SAT}_6$ be the 2SAT instance on six Boolean variables
$x_1,x_2,x_3,x_4,x_5,x$ with the following constraints:
$\overline{x_1} \vee \overline{x_2}$, $\overline{x_1} \vee
\overline{x_4}$, $x_1 \vee \overline{x_3}$, $x_1 \vee
\overline{x_5}$, $x_2 \vee \overline{x}$, $x_4 \vee x$, $x_3 \vee
\overline{x}$, $x_5 \vee x$.
\end{itemize}
\end{proof}

The following lemma is then key to proving that we have identified
all possible irreducible quantified var-elim patterns.

\begin{lemma} \label{lem:contained}
The only flat quantified irreducible patterns that do not contain any
of the patterns listed in Lemma~\ref{lem:nonVEpatterns} are contained
in BTP, invsubBTP or snake (shown in Figure~\ref{fig:VEpatterns}).
\end{lemma}

\begin{proof}
Consider a flat quantified irreducible pattern $P=\tuple{X,D,A,\cpt}$
that does not contain any of the patterns listed in
Lemma~\ref{lem:nonVEpatterns}. Thus $P$ has at most three variables,
each with domain size at most two.

We consider first the case of a 2-variable pattern $P$. By
Lemma~\ref{lem:nonVEpatterns}, $P$ does not have two non-mergeable
incompatibility edges and does not contain Z. Since $P$ is
irreducible and hence does not have any dangling assignment, we can
deduce by exhausting over all possibilities that $P$ does not have
any compatibility edge and a single incompatibility edge. Hence $P$
is contained in BTP. We can therefore assume that $P$ has exactly
three variables.

Now consider the negative sub-pattern $P^{-} =
\tuple{X,D,A,\negskel}$ where the compatibility function \negskel\ is
\cpt\ with its domain reduced to the incompatible pairs of
assignments of $P$.

Any irreducible pattern on three variables that does not contain an
incompatible pair of assignments must contain Triangle. Moreover, if
any assignment is incompatible with two other assignments then $P$
must contain either Pivot(sym) or Pivot(asym), or have two
non-mergeable incompatible edges in the same constraint. Now, since
$P$ does not contain Cycle(3), it follows that $P^{-}$ is $I_1$ or
$I_2$, as shown in Figure~\ref{fig:skeletons}.

\thicklines
\setlength{\unitlength}{1pt}
\begin{figure}
\centering

\begin{picture}(300,90)(5,0)

\put(0,0){
\begin{picture}(200,90)(-15,-13)
\put(0,20){\usebox{\varonebig}} \put(60,20){\usebox{\varonebig}}
\dashline{5}(10,40)(70,40) \put(40,-5){\makebox(0,0){$I_1$}}
\put(-6,40){\makebox(0,0){$a$}} \put(86,40){\makebox(0,0){$a'$}}
\put(10,18){\makebox(0,0){$v_0$}} \put(70,18){\makebox(0,0){$v_1$}}
\end{picture}}

\put(150,0){
\begin{picture}(95,90)(0,2)
\put(0,50){\usebox{\vartwobig}} \put(30,10){\usebox{\varonebig}}
\put(60,50){\usebox{\varonebig}} \dashline{5}(10,80)(70,70)
\dashline{5}(40,30)(10,60) \put(74,10){\makebox(0,0){$I_2$}}
\put(-6,80){\makebox(0,0){$a$}} \put(-6,60){\makebox(0,0){$b$}}
\put(86,70){\makebox(0,0){$c$}} \put(56,30){\makebox(0,0){$e$}}
\put(10,43){\makebox(0,0){$v_0$}} \put(70,48){\makebox(0,0){$v_1$}}
\put(40,8){\makebox(0,0){$v_2$}}
\end{picture}}

\end{picture}

\caption{The possible negative skeletons of var-elim patterns.}

\label{fig:skeletons}

\end{figure}
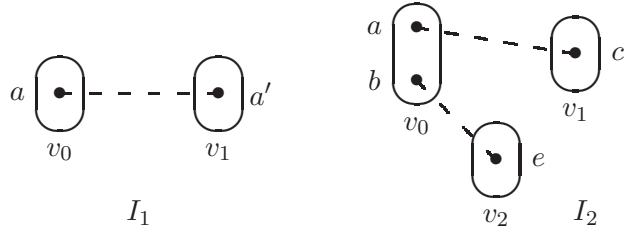

We first consider the latter case. Without loss of generality, we
assume that $b$ is compatible with $c$, to avoid $a$ and $b$ being
mergeable.

Since the domains have at most two elements, we begin by assuming
that $\mathcal{D}(v_1) = \{c,d\}$ and $\mathcal{D}(v_2) = \{e,f\}$.
In this case $a$ and $d$ must be compatible to avoid $d$ and $c$
being mergeable. Also $b$ and $f$ must be compatible to stop $e$ and
$f$ being mergeable. Now $d$ and $b$ cannot be compatible since
otherwise Z occurs in $P$. Moreover, $d$ and $e$ cannot be compatible
since otherwise XL occurs in $P$. Furthermore, $d$ and $f$ cannot be
compatible since, whichever variable is chosen for $\overline{v}(P)$,
either Kite(sym) or Kyte(asym) occurs in $P$. It follows that $d$ can
be removed as it is a dangling assignment.

Now we begin again. As before, to avoid $e$ and $f$ being mergeable
or Diamond occurring in $P$, we have that $f$ is compatible with $b$
and not compatible with $a$. To avoid Triangle occurring in $P$, $f$
cannot be compatible with $c$, which means that $f$ can be removed
since it is a dangling assignment.

So, we have $\mathcal{D}(v_1) =\{c\}$ and $\mathcal{D}(v_2) = \{e\}$.
Suppose that there is a compatibility edge between $c$ and $e$. If
the distinguished variable $\overline{v}(P)$ is $v_0$ then, whether
or not there is a compatibility edge between $a$ and $e$, the pattern
is contained in BTP. If $\overline{v}(P) = v_1$ and there is no
compatibility edge between $a$ and $e$, then the pattern is contained
in invsubBTP. If $\overline{v}(P) = v_1$ and there is a compatibility
edge between $a$ and $e$, then the pattern contains rotsubBTP. If
$\overline{v}(P) = v_2$, then the pattern contains rotsubBTP. Since
we have covered all cases in which there is a compatibility edge
between $c$ and $e$, we assume that there is no edge between $c$ and
$e$.

Whether or not there is an incompatibility edge between $a$ and $e$,
the pattern is contained in BTP if $\overline{v}(P) = v_0$, and the
pattern is contained in snake if $\overline{v}(P)$ is either $v_1$ or
$v_2$.

The final case to consider is when $P$ is a 3-variable pattern with
$P^{-} = I_1$. Any two assignments for the third variable $v_2$ could
be merged, so we can assume its domain is a singleton which we denote
by $\{a''\}$. Since $P$ is irreducible, does not contain Diamond, Z, 
Triangle, Kite(sym) or Kite(asym), we can deduce that the only
compatible pairs of assignments include $a''$. In fact, both
$\{a,a''\}$ and $\{a',a''\}$ must be compatible since $P$ is
irreducible. But then $P$ is contained in BTP if $\overline{v}(P)$ is
either $v_0$ or $v_1$, and is contained in invsubBTP if
$\overline{v}(P) = v_2$.
\end{proof}

We need the following technical lemma which shortens several proofs.

\begin{lemma}\label{lemma:trans}
If a pattern $P$ occurs in a var-elim pattern $Q$ with $|e(Q)| \leq
1$, then $P$ is also a var-elim pattern.
\end{lemma}
\begin{proof}
Suppose that $P$ occurs in the var-elim pattern $Q$ and that $|e(Q)|
\leq 1$. By transitivity of the occurrence relation, if $Q$ occurs in
a binary CSP instance $I$ (at variable $x$), then so does $P$. It
follows that if (there is an injective mapping $m: e(P) \rightarrow
\mathcal{D}(x)$ for which) $P$ does not occur (at variable $x$) in an
arc consistent binary CSP instance $I$, then (there is an injective
mapping $m': e(Q) \rightarrow \mathcal{D}(x)$ for which) $Q$ does not
occur (at variable $x$) and hence variable elimination is possible.
\end{proof}

The condition $|e(Q)| \leq 1$ is required in the statement of
Lemma~\ref{lemma:trans}, since for an instance in which
$\mathcal{D}(x)$ is a singleton, if $|e(P)| \leq 1$ and $|e(Q)|>1$
there may be an injective mapping $m: e(P) \rightarrow
\mathcal{D}(x)$ for which $P$ does not occur at $x$ but there can
clearly be no injective mapping $m': e(Q) \rightarrow
\mathcal{D}(x)$.

According to Definition~\ref{def:occurs}, a flat quantified pattern
$P$ is a sub-pattern of any existential version $Q$ of $P$ (and hence
$P$ occurs in $Q$). We state this special case of
Lemma~\ref{lemma:trans} as a corollary.

\begin{corollary}\label{cor:flatten}
Let $Q$ be an existential var-elim pattern with $|e(Q)| = 1$. If $P$
is the flattened version of pattern $Q$, corresponding to $e(P) =
\emptyset$, then $P$ is also a var-elim pattern.
\end{corollary}

The following lemma deals with the case of existential patterns $P$
with $|e(P)| > 1$.

\begin{lemma} \label{lem:e(P)>1}
No irreducible existential pattern $P$ with $|e(P)| > 1$ is a
var-elim pattern.
\end{lemma}

\begin{proof}
Let $a_1,a_2$ be two distinct assignments in $e(P)$. Since $P$ is
irreducible, $a_1$ and $a_2$ are not mergeable; so there is an
assignment $b$ such that $\tuple{b,a_1}$ is a compatibility edge and
$\tuple{b,a_2}$ is an incompatibility edge (or vice versa) in $P$.

Consider the instance $I_4^k$ (where $k = |e(P)| + 3$) on four
variables $x_1, x_2, x_3, x$ with domains $\mathcal{D}(x_1) =
\mathcal{D}(x_2) = \mathcal{D}(x_3) = \{0,1,2\}$, $\mathcal{D}(x) =
\{1,\ldots,k\}$ and the following constraints: $x_1 = 2 - x_2$, $x_1
= 2 - x_3$, $x_2 = 2 - x_3$, $(x_i \neq 1) \vee (x = i)$ ($i=1,2,3$).
$I_4^k$ has a partial solution (1,1,1) on variables $x_1,x_2,x_3$ but
has no solution. Furthermore, for any (arbitrary choice of)
injective mapping $m: e(P)
\rightarrow \mathcal{D}(x)$ which maps $e(P)$ to a subset of
$\{4,\ldots,k\}$, $P$ does not occur on $x$ since the values
$m(a_1),m(a_2) \in \{4,\ldots,k\}$ have the same compatibilities with
all assignments to other variables in $I_4^k$.

Therefore there are no irreducible var-elim patterns $P$ with $|e(P)|
> 1$.
\end{proof}

The following theorem is a direct consequence of
Theorem~\ref{thm:allVEpatterns} and Corollary~\ref{cor:flatten}
together with Lemma~\ref{lem:nonVEpatterns},
Lemma~\ref{lem:contained} and Lemma~\ref{lemma:trans}.

\begin{theorem}  \label{thm:flat}
The irreducible flat quantified patterns allowing variable
elimination in arc-consistent binary CSP instances are BTP, invsubBTP
or snake (and their irreducible sub-patterns).
\end{theorem}

We are now able to provide the characterisation for existential
patterns after a little extra work.

\begin{theorem}  \label{thm:existential}
The only irreducible existential patterns which allow variable
elimination in arc-consistent binary CSP instances are
$\exists$subBTP, $\exists$invsubBTP, $\exists$snake (and their
irreducible sub-patterns).
\end{theorem}

\begin{proof}
By Lemma~\ref{lem:e(P)>1} we only need to consider patterns $P$ with
$|e(P)| = 1$.

We know from Theorem~\ref{thm:allVEpatterns} that $\exists$subBTP,
$\exists$invsubBTP, $\exists$snake are var-elim patterns.

Theorem~\ref{thm:flat} and  Corollary~\ref{cor:flatten} show that
when we flatten an existential var-elim pattern then the resulting
flat quantified pattern is contained in BTP, invsubBTP or snake.

In the case of invsubBTP and snake, the existential versions of these
patterns are var-elim patterns and so there is nothing left to prove.
So we only need to consider quantified patterns which flatten into
sub-patterns of BTP.

Let $\exists$BTP denote the existential version $Q$ of BTP such that
$|e(Q)|=1$. By symmetry, $\exists$BTP is unique. The only remaining
case is when $P$ is an irreducible sub-pattern of $\exists$BTP with
$|e(P)|=1$. By a straightforward exhaustive case analysis, we find
that, in this case, either $P$ contains V($+ -$) or Triangle(asym) or
$P$ is a sub-pattern of $\exists$subBTP. The result then follows by
Lemma~\ref{lem:nonVEpatterns} and Lemma~\ref{lemma:trans}.
\end{proof}

Combining Theorem~\ref{thm:flat} and Theorem~\ref{thm:existential},
we obtain the characterisation of irreducible quantified var-elim
patterns.

\begin{theorem} \label{thm:dichotomy}
The only irreducible quantified patterns which allow variable
elimination in arc-consistent binary CSP instances are $BTP$,
$\exists$subBTP, $\exists$invsubBTP, $\exists$snake (and their
irreducible sub-patterns).
\end{theorem}

It is easy to see that variable elimination cannot destroy arc
consistency. Hence there is no need to re-establish arc consistency
after variable eliminations. Furthermore, the result of applying our
var-elim rules until convergence is unique; variable eliminations may
lead to new variable eliminations but cannot introduce patterns and
hence cannot invalidate applications of our var-elim rules.

\section{Value elimination patterns}

We now consider when forbidding a pattern can allow the elimination
of values from domains rather than the elimination of variables.
Value-elimination is at the heart of the simplification operations
employed by constraint solvers during preprocessing or during search.
In current solvers such eliminations are based almost exclusively on
consistency operations: a value is eliminated from the domain of a
variable if this assignment can be shown to be inconsistent (in the
sense that it cannot be part of any solution). Another
value-elimination operation which can be applied is neighbourhood
substitutability which allows the elimination of certain assignments
which are unnecessary for determining the satisfiability of the
instance. Neighbourhood substitutability  can be described by means
of the pattern shown in Figure~\ref{fig:ns}. If in a binary CSP
instance $I$, there are two assignments $a,b$ for the same variable
$x$ such that this pattern does not occur (meaning that $a$ is
consistent with all assignments with which $b$ is consistent), then
the assignment $b$ can be eliminated. This is because in any solution
containing $b$, simply replacing $b$ by $a$ produces another
solution.

It is worth noting that even when all solutions are required,
neighbourhood substitutability can still be applied since all
solutions to the original instance can be recovered from the set of
solutions to the reduced instance in time which is linear in the
total number of solutions and polynomial in the size of the
instance~\cite{Cooper97:neighbourhood}.

\setlength{\unitlength}{1pt} 
\begin{figure}
\centering

\begin{picture}(102,70)(2,40)

\put(0,50){\usebox{\varone}} \put(60,50){\usebox{\vartwo}}
\put(10,70){\line(6,-1){60}} \dashline{5}(10,70)(70,80)
\put(70,40){\makebox(0,0){$x$}} \put(89,80){\makebox(0,0){$a$}}
\put(84,55){\framebox(10,10){$b$}}
\put(40,100){\makebox(0,0){$\exists x \exists b \exists a$}}
\end{picture}

\caption{A value elimination pattern corresponding to neighbourhood
substitution.}

\label{fig:ns}

\end{figure}
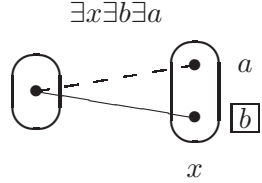
\setlength{\unitlength}{1pt}

\begin{definition}
We say that a value $b \in \mathcal{D}(x)$ \emph{can be eliminated}
from an instance $I$ if the instance $I'$ in which the assignment $b$
has been deleted from $\mathcal{D}(x)$ is satisfiable if and only if
$I$ is satisfiable.
\end{definition}

\begin{definition} \label{def:value-elimination}
An existential pattern $P$ with a distinguished value
$\overline{val}(P)$ is a \emph{value elimination pattern (val-elim
pattern)} if in all arc-consistent instances $I$, whenever the
pattern does not occur at a variable $x$ in $I$ for at least one
injective value mapping $m$, the value $m(\overline{val}(P))$ can be
eliminated from $\mathcal{D}(x)$ in $I$.
\end{definition}

An obvious question is which patterns allow value elimination while
preserving satisfiability? The following theorem gives three
existential patterns which provide strict generalisations of
neighbourhood substitutability since in each case the pattern of
Figure~\ref{fig:ns} is a sub-pattern. In each of the patterns $P$ in
Figure~\ref{fig:ns} and Figure~\ref{fig:value-elimination-patterns},
the value that can be eliminated $\overline{val}(P)$ is the value $b$
surrounded by a small box.

\begin{theorem} \label{thm:all-val-elim -patterns}
The existential patterns shown in
Figure~\ref{fig:value-elimination-patterns}, namely $\exists$2snake,
$\exists$2invsubBTP and $\exists$2triangle, are each val-elim patterns.
\end{theorem}

\setlength{\unitlength}{1pt} 
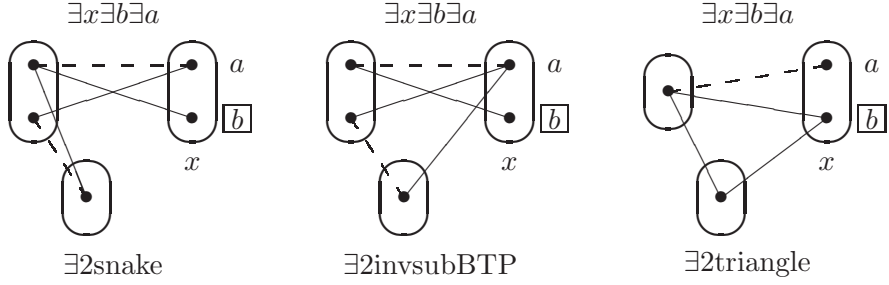
\begin{figure}
\centering

\begin{picture}(342,120)(2,0)

\put(0,0){
\begin{picture}(102,120)(0,0)
\put(0,50){\usebox{\vartwo}} \put(20,10){\usebox{\varone}}
\put(60,50){\usebox{\vartwo}}  \put(10,60){\line(3,1){60}}
\put(10,80){\line(3,-1){60}} \put(10,80){\line(2,-5){20}}
\dashline{5}(10,80)(70,80) \dashline{5}(30,30)(10,60)
\put(70,43){\makebox(0,0){$x$}} \put(87,80){\makebox(0,0){$a$}}
\put(82,55){\framebox(10,10){$b$}}
\put(40,100){\makebox(0,0){$\exists x \exists b \exists a $}}
\put(40,5){\makebox(0,0){$\exists$2snake}}
\end{picture}}

\put(120,0){
\begin{picture}(102,120)(0,0)
\put(0,50){\usebox{\vartwo}} \put(20,10){\usebox{\varone}}
\put(60,50){\usebox{\vartwo}}  \put(10,60){\line(3,1){60}}
\put(10,80){\line(3,-1){60}} \put(30,30){\line(4,5){40}}
\dashline{5}(10,80)(70,80) \dashline{5}(30,30)(10,60)
\put(70,43){\makebox(0,0){$x$}} \put(87,80){\makebox(0,0){$a$}}
\put(82,55){\framebox(10,10){$b$}}
\put(40,100){\makebox(0,0){$\exists x \exists b \exists a $}}
\put(40,5){\makebox(0,0){$\exists$2invsubBTP}}
\end{picture}}

\put(240,0){
\begin{picture}(102,120)(0,0)
\put(0,50){\usebox{\varone}} \put(20,10){\usebox{\varone}}
\put(60,50){\usebox{\vartwo}}  \put(10,70){\line(6,-1){60}}
\put(10,70){\line(1,-2){20}} \put(30,30){\line(4,3){40}}
\dashline{5}(10,70)(70,80) \put(70,43){\makebox(0,0){$x$}}
\put(87,80){\makebox(0,0){$a$}} \put(82,55){\framebox(10,10){$b$}}
\put(40,100){\makebox(0,0){$\exists x \exists b \exists a $}}
\put(40,5){\makebox(0,0){$\exists$2triangle}}
\end{picture}}

\end{picture}

\caption{Three val-elim patterns.}

\label{fig:value-elimination-patterns}

\end{figure}
\setlength{\unitlength}{1pt}

\begin{proof}
We first show that the result holds for instances $I$ with at most two
variables. Let $x$ be a variable of $I$. For
$|\mathcal{D}(x)|
> 1$, there is clearly an injective mapping $m: e(P) \rightarrow
\mathcal{D}(x)$ for which none of the patterns $P$ shown in
Figure~\ref{fig:value-elimination-patterns} occur since they all have
three variables. But, we can always eliminate all but one value in
$\mathcal{D}(x)$ without destroying satisfiability, since by arc
consistency the remaining value is necessarily part of a solution.
For $|\mathcal{D}(x)| \leq 1$, there can be no injective mapping $m:
e(P) \rightarrow \mathcal{D}(x)$ since $|e(P) = 2|$ and hence there
is nothing to prove. In the rest of the proof we therefore only need
to consider instances $I = \tuple{X,D,A,\cpt}$ with at least three
variables. We will prove the result for each of the three patterns
one by one.

We consider first $\exists$2snake. Suppose that for a variable $x$
and values $a,b \in \mathcal{D}(x)$, the pattern $\exists$2snake does
not occur. Let $I'$ be identical to $I$ except that value $b$ has
been eliminated from $\mathcal{D}(x)$. Suppose that $s$ is a solution
to $I$ with $s(x) = b$. It suffices to show that $I'$ has a solution.
Let $Y$ ($\overline{Y}$) be the set of variables $z \in X \setminus
\{x\}$ such that $\cpt(\tuple{z,s(z)},\tuple{x,a})$ $=$ {\tt TRUE}
({\tt FALSE}). By arc consistency, there are assignments
$\tuple{z,t(z)}$ for all $z \in \overline{Y}$ which are compatible
with $\tuple{x,a}$. Let $z \in \overline{Y}$ and $y \in X \setminus
\{x,z\}$. Since $s$ is a solution with $s(x) = b$,
$\cpt(\tuple{y,s(y)}, \tuple{z,s(z)})=$ $\cpt(\tuple{x,b},
\tuple{z,s(z)})=$ {\tt TRUE}. Since 
$\exists$2snake does not occur on $\{\tuple{x,a}, \tuple{x,b},
\tuple{z,s(z)}, \tuple{z,t(z)}, \tuple{y,s(y)}\}$, we can deduce that
$\cpt(\tuple{z,t(z)}, \tuple{y,s(y)})=$ {\tt TRUE}. In particular, we
have $\cpt(\tuple{y,s(y)}, \tuple{z,t(z)})=$ {\tt TRUE} for all $y
\neq z \in \overline{Y}$. Then, since $\exists$2snake does not occur
on the assignments $\{\tuple{x,a}, \tuple{x,b}, \tuple{y,s(y)},
\tuple{y,t(y)}, \tuple{z,t(z)}\}$, we can deduce
$\cpt(\tuple{z,t(z)}, \tuple{y,t(y)})=$ {\tt TRUE}. Hence the
assignments $\tuple{z,t(z)}$ ($z \in \overline{Y}$) are compatible
between themselves, are all compatible with all $\tuple{y,s(y)}$ ($y
\in Y$) and with $\tuple{x,a}$. Thus, $s'$ is a solution to $I'$,
where
\begin{equation*}
s'(v)=
\begin{cases}
a & \text{if $v = x$,}\\
s(v) & \text{if $v \in Y$,}\\
t(v) & \text{otherwise.}
\end{cases}
\end{equation*}

We now consider $\exists$2invsubBTP. Suppose that for a variable $x$
and values $a,b \in \mathcal{D}(x)$ in an instance $I$, the pattern
$\exists$2invsubBTP does not occur. Let $I'$ be identical to $I$
except that value $b$ has been eliminated from $\mathcal{D}(x)$.
Suppose that $s$ is a solution to $I$ with $s(x) = b$ and again let
$Y$ ($\overline{Y}$) be the set of variables $z \in X \setminus
\{x\}$ such that $\cpt(\tuple{z,s(z)},\tuple{x,a})$ $=$ {\tt TRUE}
({\tt FALSE}). By arc consistency, for each $z \in \overline{Y}$,
there is an assignment $\tuple{z,t(z)}$ which is compatible with
$\tuple{x,a}$. Let $s'$ be defined as above. Consider $v \in X \setminus \{x\}$.
We know that $\cpt(\tuple{x,a}, \tuple{v,s'(v)})=$ {\tt TRUE}. Let $z
\in \overline{Y}$. Since the pattern $\exists$2invsubBTP does not
occur on $\{\tuple{x,a}, \tuple{x,b}, \tuple{z,s(z)}, \tuple{z,t(z)},
\tuple{v,s'(v)}\}$, we can deduce that $\cpt(\tuple{z,t(z)},
\tuple{v,s'(v)})=$ {\tt TRUE}. It follows that $s'$ is a solution to
$I'$.

Finally, we consider $\exists$2triangle. Suppose that in an instance
$I$, for values $a,b \in \mathcal{D}(x)$, the pattern
$\exists$2triangle does not occur. Let $I'$ be identical to $I$
except that value $b$ has been eliminated from $\mathcal{D}(x)$.
Suppose that $s$ is a solution to $I$ with $s(x) = b$. Then
$\tuple{x,a}$ must be compatible with all assignments
$\tuple{y,s(y)}$ ($y \in X \setminus \{x\}$), otherwise the pattern
$\exists$2triangle would occur on $\{\tuple{x,a}, \tuple{x,b},
\tuple{y,s(y)}, \tuple{z,s(z)}\}$ for all $z \in X \setminus
\{x,y\}$. It follows that $s''$ is a solution to $I'$, where
\begin{equation*}
s''(v)=
\begin{cases}
a & \text{if $v = x$,}\\
s(v) & \text{otherwise.}
\end{cases}
\end{equation*}
\end{proof}

\begin{example}
Consider a CSP instance corresponding to a problem of colouring a
complete graph on four vertices. The colours assigned to the four
vertices are represented by variables $x_1,x_2,x_3,x_4$ whose domains
are, respectively, $\{0,1,2,3\}$, $\{0,1\}$, $\{0,2\}$, $\{0,3\}$.
Notice that the instance is arc consistent and no eliminations are
possible by neighbourhood substitution. However, the value 1 can be
eliminated from the domain of $x_1$ since for the mapping $a \mapsto
0$, $b \mapsto 1$, the pattern $\exists$2snake does not occur on
$x_1$. The values 2 and 3 can also be eliminated from the domain of
$x_1$ for the same reason. After applying arc consistency to the
resulting instance, all domains are singletons.
\end{example}

\begin{example}
Consider the arc-consistent instance on three Boolean variables
$x,y,z$ and with the constraints $z \vee \neg x$, $z \vee y$, $\neg y
\vee \neg x$. In this instance we can eliminate the assignment
$\tuple{x,0}$ since $\exists$2invsubBTP does not occur on variable
$x$ for the mapping $a \mapsto 1$, $b \mapsto 0$. The assignments
$\tuple{y,1}$ and $\tuple{z,0}$ then have no support at $x$ and hence
can be eliminated by arc consistency, leaving an instance in which
all domains are singletons.
\end{example}

\begin{example}
Consider the arc-consistent CSP instance corresponding to a graph
colouring problem on a complete graph on three vertices in which the
domains of variables $x_1,x_2,x_3$ are each $\{0,1\}$. Again, no
eliminations are possible by neighbourhood substitution. However, the
value 1 can be eliminated from the domain of $x_1$ since for the
mapping $a \mapsto 0$, $b \mapsto 1$, the pattern $\exists$2triangle
does not occur on $x_1$. Applying arc consistency then leads to an
empty domain from which we can deduce that the original instance was
unsatisfiable.
\end{example}

Neighbourhood substitution cannot destroy arc
consistency~\cite{Cooper97:neighbourhood}, but eliminating a value by
a val-elim pattern can provoke new eliminations by arc consistency,
as we have seen in the above examples.

The result of applying a sequence of neighbourhood substitution
eliminations until convergence is unique modulo
isomorphism~\cite{Cooper97:neighbourhood}. This is not true for the
result of eliminating domain elements by val-elim patterns, as the
following example demonstrates.

\begin{example}
Consider the CSP instance on three variables $x_1,x_2,x_3$, each with
domain $\{0,1,2\}$, and with the following constraints: $(x_1 \neq 2)
\vee (x_2 \neq 2)$, $(x_1,x_3) \in R$, $(x_2,x_3) \in R$, where $R$
is the relation $\{(0,0), (0,2), (1,1), (2,1)$, $(2,2)\}$. We can
eliminate the assignment $\tuple{x_3,0}$ since $\exists$2snake does
not occur on $x_3$ with the value $a$ mapping to $2$ and $b$ to $0$.
But then in the resulting arc-consistent instance, no more
eliminations are possible by any of the val-elim patterns shown in
Figure~\ref{fig:value-elimination-patterns}. However, in the original
instance we could have eliminated the assignment $\tuple{x_3,1}$
since $\exists$2snake does not occur on $x_3$ with the value $a$
mapping to $0$ and $b$ to $1$. Then we can successively eliminate
$\tuple{x_1,1}$, $\tuple{x_2,1}$ by arc consistency and then
$\tuple{x_1,2}$, $\tuple{x_2,2}$, $\tuple{x_3,0}$ by $\exists$2snake.
In the resulting instance all domains are singletons. Thus, for this
instance there are two convergent sequences of value eliminations
which produce non-isomorphic instances.
\end{example}

It is clear that variable elimination by our var-elim rules can
provoke new value eliminations by our val-elim rules. Value
elimination may provoke new variable eliminations, but may also
invalidate a variable elimination if the value eliminated (or one of
the values eliminated by subsequent arc consistency operations) is
the only value on which an existential var-elim pattern does not
occur. Thus, to maximize reductions, variable eliminations should
always be performed before value eliminations.

\section{Characterisation of value elimination patterns}

As with existential variable-elimination patterns, we can give a
dichotomy for irreducible existential val-elim patterns. We first
require the following lemma which shows that many patterns, including
those illustrated in Figure~\ref{fig:not-value-elimination} (along
with the patterns Z and Diamond shown in
Figure~\ref{fig:nonVEpatterns}), cannot be contained in val-elim
patterns. In Figure~\ref{fig:not-value-elimination}, each of the
patterns I($-$), L($+-$), triangle1, triangle2, $\exists$Kite,
$\exists$Kite(asym) and $\exists$Kite1 has a distinguished value $b =
\overline{val}(P)$ which is highlighted in the figure by placing the
value in a small box.

\setlength{\unitlength}{1pt} 
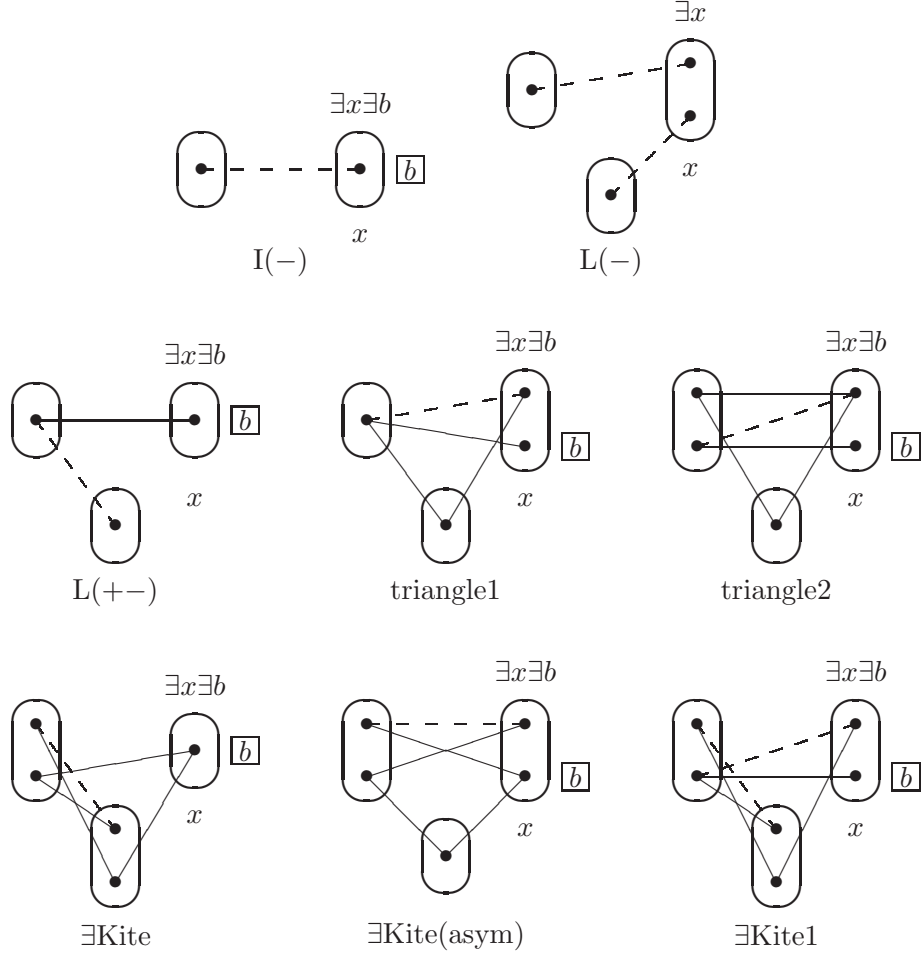
\begin{figure}
\centering

\begin{picture}(354,360)(10,0)

\put(62.5,250){
\begin{picture}(104,90)(-10,-20)
\put(0,10){\usebox{\varone}} \put(60,10){\usebox{\varone}}
\dashline{5}(10,30)(70,30) \put(70,5){\makebox(0,0){$x$}}
\put(84,25){\framebox(10,10){$b$}}
\put(70,55){\makebox(0,0){$\exists x \exists b$}}
\put(40,-5){\makebox(0,0){I($-$)}}
\end{picture}}


\put(187.5,250){
\begin{picture}(104,110)(-10,-20)
\put(0,40){\usebox{\varone}} \put(60,40){\usebox{\vartwo}}
\put(30,0){\usebox{\varone}} \dashline{5}(10,60)(70,70)
\dashline{5}(40,20)(70,50) \put(70,30){\makebox(0,0){$x$}}
\put(70,90){\makebox(0,0){$\exists x$}}
\put(40,-5){\makebox(0,0){L($-$)}}
\end{picture}}


\put(0,125){
\begin{picture}(104,110)(-10,-20)
\put(0,40){\usebox{\varone}} \put(60,40){\usebox{\varone}}
\put(30,0){\usebox{\varone}} \put(10,60){\line(1,0){60}}
\dashline{5}(10,60)(40,20) \put(70,30){\makebox(0,0){$x$}}
\put(84,55){\framebox(10,10){$b$}} \put(70,85){\makebox(0,0){$\exists
x \exists b$}} \put(40,-5){\makebox(0,0){L($+-$)}}
\end{picture}}

\put(125,125){
\begin{picture}(104,110)(-10,-20)
\put(0,40){\usebox{\varone}} \put(30,0){\usebox{\varone}}
\put(60,40){\usebox{\vartwo}} \put(10,60){\line(6,-1){60}}
\dashline{5}(10,60)(70,70) \put(10,60){\line(3,-4){30}}
\put(40,20){\line(3,5){30}} \put(70,30){\makebox(0,0){$x$}}
\put(84,45){\framebox(10,10){$b$}} \put(70,90){\makebox(0,0){$\exists
x \exists b$}} \put(40,-5){\makebox(0,0){triangle1}}
\end{picture}}

\put(250,125){
\begin{picture}(104,110)(-10,-20)
\put(0,40){\usebox{\vartwo}} \put(30,0){\usebox{\varone}}
\put(60,40){\usebox{\vartwo}} \put(10,70){\line(1,0){60}}
\put(10,50){\line(1,0){60}} \dashline{5}(10,50)(70,70)
\put(10,70){\line(3,-5){30}} \put(40,20){\line(3,5){30}}
\put(70,30){\makebox(0,0){$x$}} \put(84,45){\framebox(10,10){$b$}}
\put(70,90){\makebox(0,0){$\exists x \exists b$}}
\put(40,-5){\makebox(0,0){triangle2}}
\end{picture}}

\put(0,0){
\begin{picture}(104,110)(-10,-20)
\put(0,40){\usebox{\vartwo}} \put(60,40){\usebox{\varone}}
\put(30,0){\usebox{\vartwo}} \put(10,50){\line(6,1){60}}
\dashline{5}(10,70)(40,30) \put(10,70){\line(1,-2){30}}
\put(10,50){\line(3,-2){30}} \put(40,10){\line(3,5){30}}
\put(70,35){\makebox(0,0){$x$}} \put(84,55){\framebox(10,10){$b$}}
\put(70,85){\makebox(0,0){$\exists x \exists b$}}
\put(40,-10){\makebox(0,0){$\exists$Kite}}
\end{picture}}

\put(125,0){
\begin{picture}(104,110)(-10,-20)
\put(0,40){\usebox{\vartwo}} \put(60,40){\usebox{\vartwo}}
\put(30,0){\usebox{\varone}} \put(10,50){\line(3,1){60}}
\dashline{5}(10,70)(70,70) \put(10,70){\line(3,-1){60}}
\put(10,50){\line(1,-1){30}} \put(40,20){\line(1,1){30}}
\put(70,30){\makebox(0,0){$x$}} \put(84,45){\framebox(10,10){$b$}}
\put(70,90){\makebox(0,0){$\exists x \exists b$}}
\put(40,-10){\makebox(0,0){$\exists$Kite(asym)}}
\end{picture}}

\put(250,0){
\begin{picture}(104,110)(-10,-20)
\put(0,40){\usebox{\vartwo}} \put(30,0){\usebox{\vartwo}}
\put(60,40){\usebox{\vartwo}} \put(10,50){\line(1,0){60}}
\put(10,50){\line(3,-2){30}} \dashline{5}(10,50)(70,70)
\dashline{5}(10,70)(40,30) \put(10,70){\line(1,-2){30}}
\put(40,10){\line(1,2){30}} \put(70,30){\makebox(0,0){$x$}}
\put(84,45){\framebox(10,10){$b$}} \put(70,90){\makebox(0,0){$\exists
x \exists b$}} \put(40,-10){\makebox(0,0){$\exists$Kite1}}
\end{picture}}

\end{picture}

\caption{Patterns which do not allow value elimination.}

\label{fig:not-value-elimination}

\end{figure}
\setlength{\unitlength}{1pt}

\begin{lemma} \label{lem:not-value-elimination}
None of the following existential patterns $P$ (with a distinguished
value $\overline{val}(P)$) allow value elimination in arc-consistent
binary CSP instances: any pattern on strictly more than three
variables, any pattern with three non-mergeable values for the same
variable $v \neq \overline{v}(P)$, any pattern with two non-mergeable
incompatibility edges in the same constraint, any pattern containing
any of Z, Diamond, I($-$), L($-$), L($+-$), triangle1, triangle2,
$\exists$Kite, $\exists$Kite(asym) or $\exists$Kite1.
\end{lemma}

\begin{proof}
For each pattern we exhibit a binary arc-consistent CSP instance $I$
and a value $b$ for a variable $x$ in $I$ such that
\begin{itemize}
\item $I$ has a solution which includes the assignment $\tuple{x,b}$;
\item $I$ has no solution if the value $b$ is deleted from $\mathcal{D}(x)$;
\item $I$ does not contain the given pattern $P$ on $x$ with
$\overline{val}(P)$ mapping to $b$.
\end{itemize}
By definition, any such instance is sufficient to prove that the
pattern $P$ is not a val-elim pattern. Since in existential patterns
$P$, the set $e(P)$ may be of arbitrary size, we have to give generic
instances in which the size of the domain of $x$ is arbitrarily
large.

\begin{itemize}
\item For any pattern $P$ which either contains I($-$) or has strictly more than three
variables or with three non-mergeable values for the same variable $v
\neq \overline{v}(P)$.

Let $I^{SAT}_{3}$ be the arc-consistent instance on three variables
$x_1, x_2, x$ with domains $\mathcal{D}(x_1) = \mathcal{D}(x_2) =
\{0,1\}$, $\mathcal{D}(x) = \{0,\ldots,k\}$ and with the following
constraints: $\overline{x_1} \vee \overline{x_2}$, $x_1 \vee (x=0)$,
$x_2 \vee (x=0)$. $I^{SAT}_{3}$ has a solution $\tuple{0,0,0}$ which
includes the assignment $\tuple{x,0}$, has no solution if this
assignment is eliminated, and does not contain $P$ on $x$ with
$\overline{val}(P)$ mapping to 0.

\item For any pattern $P$ which has two non-mergeable incompatibility edges in
the same constraint.

Let $I^{SAT}_{2k+1}$ be the arc-consistent instance on $2k+1$
variables $x_1, \ldots, x_{2k}, x$ with domains $\mathcal{D}(x_1) =
\ldots = \mathcal{D}(x_{2k}) = \{0,1\}$, $\mathcal{D}(x) =
\{0,\ldots,k\}$ and with the following constraints for each
$i=1,\ldots,k$: $\overline{x_{2i-1}} \vee \overline{x_{2i}}$,
$x_{2i-1} \vee (x \neq i)$, $x_{2i} \vee (x \neq i)$.
$I^{SAT}_{2k+1}$ has a solution $\tuple{0,\ldots,0}$ which includes
the assignment $\tuple{x,0}$, has no solution if this assignment is
eliminated, and does not contain $P$ on $x$ with $\overline{val}(P)$
mapping to 0.

\item For any pattern $P$ which contains L($+-$), triangle2 or $\exists$Kite1.

Let $I_3$ be the arc-consistent instance on three variables $x_1,
x_2, x$ with domains $\mathcal{D}(x_1) = \mathcal{D}(x_2) =
\mathcal{D}(x) = \{0,\ldots,k\}$ and with the following constraints:
$(x_1 = 0) \vee (x_2 = 0)$, $x_1 = x$, $x_2 = x$. $I_3$ has a
solution $\tuple{0,0,0}$ which includes the assignment $\tuple{x,0}$,
has no solution if this assignment is eliminated, and does not
contain $P$ on $x$ with $\overline{val}(P)$ mapping to 0.

\item For any pattern $P$ which contains L($-$).

Let $I_{3+}$ be the arc-consistent instance on four variables $x_1,
x_2, x_3, x$ with domains $\mathcal{D}(x_1) = \mathcal{D}(x_2) =
\mathcal{D}(x_3) = \mathcal{D}(x) = \{0,\ldots,k\}$ and with the
following constraints: $(x_1 = 0) \vee (x_2 = 0)$, $x_1 = x_3$, $x_2
= x_3$, $x_3 = x$. $I_{3+}$ has a solution $\tuple{0,0,0,0}$ which
includes the assignment $\tuple{x,0}$, has no solution if this
assignment is eliminated, and does not contain $P$ on $x$.

\item For any pattern $P$ which contains triangle1,
$\exists$Kite, $\exists$Kite(asym), Diamond or Z.

Let $I_{3}^{2k}$ be the arc-consistent instance on three variables
$x_1, x_2, x$ each with domain $\{0,\ldots,2k\}$ and with the
following constraints: $x_1 = 2k - x_2$, $x_1 = x$, $x_2 = x$.
$I_{3}^{2k}$ has a solution $\tuple{k,k,k}$ which includes the
assignment $\tuple{x,k}$, has no solution if this assignment is
eliminated, and does not contain $P$ on $x$ with $\overline{val}(P)$
mapping to $k$.
\end{itemize}
\end{proof}

We can now characterise those irreducible existential patterns which
allow value elimination and hence generalise neighbourhood
substitution.

\begin{theorem}
The only irreducible existential patterns which allow value
elimination in arc-consistent binary CSP instances are
$\exists$2snake, $\exists$2invsubBTP and $\exists$2triangle (and
their irreducible sub-patterns).
\end{theorem}

\begin{proof}
Let $P$ be an irreducible existential pattern which allows value
elimination in arc-consistent binary CSP instances. Since
$\overline{val}(P)$ is necessarily defined in a val-elim pattern and
belongs to $e(P)$, we need to consider three different cases:
\begin{enumerate}
\item $|e(P)| = 1$, and hence $e(P) = \{\overline{val}(P)\}$,
\item $|e(P)| = 2$,
\item $|e(P)| > 2$.
\end{enumerate}
\paragraph{Case $|e(P)| = 1$:} \ Consider the CSP instance $I_2$ consisting of
only two variables $x_1,x_2$, each with a singleton domain $\{0\}$
together with the constraint $x_1 = x_2$. Trivially, the value 0
cannot be eliminated from the domain of $x_1$ without changing the
satisfiability of the instance. Any existential pattern $P$ with
$|e(P)| = 1$ containing more than two variables or at least one
incompatibility edge does not occur in $I_2$ on $x_1$ with
$\overline{val}(P)$ mapping to 0, and hence cannot be a val-elim
pattern. There is no 2-variable irreducible existential pattern which
contains only compatibility edges. Therefore, the only irreducible
val-elim pattern $P$ with $|e(P)| = 1$ is the trivial pattern with no
edges (which is a sub-pattern of $\exists$2snake, for example).

\paragraph{Case $|e(P)| = 2$:} \ Let $e(P) = \{a,b\}$ where
$b = \overline{val}(P)$. We know from
Lemma~\ref{lem:not-value-elimination} that I($-$) cannot be contained
in a val-elim pattern. We can therefore deduce that the assignment
$\tuple{\overline{v}(P),b}$ can only belong to compatibility edges in
$P$. Since $P$ is irreducible, we can deduce that $P$ must contain
the neighbourhood substitution pattern shown in Figure~\ref{fig:ns},
otherwise $a$ and $b$ could be merged. By
Lemma~\ref{lem:not-value-elimination}, $P$ does not contain more than
three variables, does not contain L($-$) or L($+-$) and does not contain more
than one incompatibility edge per constraint. It follows that $P$
contains at most two incompatibility edges. The only extension of the
neighbourhood substitution pattern (shown in Figure~\ref{fig:ns})
containing only one incompatibility edge and containing none of
I($-$), Z, Diamond, triangle1, triangle2 or $\exists$Kite(asym) is
the val-elim pattern $\exists$2triangle. The only extensions of the
neighbourhood substitution pattern containing exactly two
incompatibility edges and containing none of I($-$), L($-$), L($+-$),
Z, Diamond, triangle1, $\exists$Kite or $\exists$Kite1 are the
val-elim patterns $\exists$2snake and $\exists$2invsubBTP. Hence, the
only irreducible val-elim patterns $P$ with $|e(P)| = 2$ are
$\exists$2snake, $\exists$2invsubBTP and $\exists$2triangle (and
their irreducible sub-patterns).

\paragraph{Case $|e(P)| > 2$:} \ Let $a_1,a_2,a_3$ be three distinct
values in $e(P)$ and, for $i=1,2,3$, let $q_i$ denote the assignment
$\tuple{\overline{v}(P),a_i}$ . Since $P$ is irreducible, for all
$i,j$ such that $1 \leq i < j \leq 3$, $a_i$ and $a_j$ are not
mergeable; so there is an assignment $p_{ij}$ such that
$\tuple{p_{ij},q_i}$ is a compatibility edge and $\tuple{p_{ij},q_j}$
is an incompatibility edge (or vice versa) in $P$. By
Lemma~\ref{lem:not-value-elimination}, $P$ has at most three
variables, including $\overline{v}(P)$. So two of $p_{12}, p_{13},
p_{23}$ are assignments to the same variable. Without loss of
generality, suppose that $p_{12}, p_{13}$ are assignments to the same
variable $y \neq \overline{v}(P)$. By
Lemma~\ref{lem:not-value-elimination}, $P$ has at most one
incompatibility edge in each constraint. It follows that $p_{12} =
p_{13}$, with $\tuple{p_{12},q_1}$ an incompatibility edge and
$\tuple{p_{12},q_2}$, $\tuple{p_{12},q_3}$ compatibility edges. It
also follows that $p_{23}$ must be an assignment to a distinct
variable $z \notin \{y,\overline{v}(P)\}$. But then $P$ contains
Diamond on $p_{12}$, $q_2$, $q_3$, $p_{23}$ and so, by
Lemma~\ref{lem:not-value-elimination}, cannot be a val-elim pattern.
Therefore there are no irreducible val-elim patterns $P$ with $|e(P)|
> 2$.
\end{proof}

\section{Recovering one or all solutions after eliminations}

The binary CSP has diverse applications. In some applications it is
only the satisfiability of the instance which is of interest. For
example, in optimal planning, to determine whether an action $a$
among a set of available actions $A$ is indispensable (i.e. that it
is present in all solution-plans) we need to determine the
satisfiability of a binary CSP representing the same planning problem
using the set of actions $A \setminus
\{a\}$~\cite{cooper:planning2011}. The variable and value elimination
rules presented in this paper are directly applicable to such
problems.

Nonetheless, in most applications, the final aim is to find one or
all solutions. In many planning, scheduling and configuration
problems, the aim is often to find just one solution which satisfies
all the constraints. In other application areas, such as
fault-diagnosis~\cite{DBLP:journals/dam/WilliamsR07} or the
interpretation of ambiguous pictures~\cite{cooper:prl1988}, it is
important to find all solutions or a representation of all solutions
from which it is possible to extract in polynomial time a solution
satisfying certain criteria. For example, the on-line configuration
of a product (such as a car) by a user can be rendered tractable by
the off-line compilation of all solutions into some appropriate
compact
form~\cite{DBLP:journals/ai/AmilhastreFM02,DBLP:conf/ictai/AmilhastreFNP12}.
We will therefore study in this section whether it is possible to
efficiently recover one or all solutions to a binary CSP instance
after elimination of variables and/or values by our rules. We will
show that the efficient recovery of one solution is always possible,
but that only some of our rules allow the efficient recovery of all
solutions.

The elimination of a variable cannot destroy arc consistency, but the
elimination of a value may do so. Throughout this section we assume
that the elimination of an assignment by applying a value-elimination
rule is necessarily immediately followed by the re-establishment of
arc consistency.

\begin{proposition} \label{prop:one-solution}
Let $I$ be an arc-consistent binary CSP instance and let $s$ be a
solution to the instance obtained after applying a sequence $\sigma$
of variable and value elimination operations (defined by irreducible
quantified var-elim or val-elim patterns in the sense of
Definition~\ref{def:VE} and Definition~\ref{def:value-elimination}).
Then a solution to $I$ can be found from $(s,\sigma)$ in $O(cd)$
time, where $c$ is the number of non-trivial constraints and $d$ the
maximum domain size in $I$.
\end{proposition}

\begin{proof}
Since value elimination does not modify constraints, any solution for
an instance obtained from $I$ by value eliminations is also a
solution to $I$. We therefore only need to consider the case of
var-elim rules. We identified the irreducible quantified var-elim
patterns in Theorem~\ref{thm:dichotomy} as (irreducible sub-patterns
of) $BTP$, $\exists$subBTP, $\exists$invsubBTP or $\exists$snake. We
only need to prove the proposition for these four patterns since
absence of any sub-pattern implies absence of the pattern itself.

Consider a single variable elimination operation consisting in
eliminating variable $x$. Let $c_x$ denote the number of constraints
whose scope includes $x$. If $x$ has been eliminated due to the
absence of BTP, then it is known that any solution $s$ of the reduced
instance can be extended to a solution of $I$~\cite{Cooper10:BTP}.
The proof of Theorem~\ref{thm:allVEpatterns} showed that this is also
true in the case that $x$ has been eliminated due to the absence of
$\exists$subBTP. In order to extend $s$ to a solution of $I$, it
suffices to test each of the elements of $\mathcal{D}(x)$ in turn
against each of the $c_x$ constraints. This can be achieved in $O(c_x
d)$ time.

For the two other patterns ($\exists$invsubBTP and $\exists$snake),
the proof of Theorem~\ref{thm:allVEpatterns} actually provides an
algorithm to recover a solution $s'$ of $I$ from the solution $s$ of
the reduced instance, via the calculation of the variable set
$\overline{Y}$ and the assignments $t(z)$ ($z \in \overline{Y}$).
Again this can be achieved in $O(c_x d)$ time.

Summing the $O(c_x d)$ complexity of recovering a solution to the
instance in which a variable $x$ is reinstated, over all eliminated
variables $x$, gives a total complexity of $O(cd)$, as claimed.
\end{proof}

\begin{proposition} \label{prop:all-solutions}
Let $I$ be an arc-consistent binary CSP instance and let $S$ be the
set of all solutions to the instance obtained after applying a
sequence $\sigma$ of operations given by the var-elim patterns BTP,
$\exists$subBTP and the val-elim pattern $\exists$2triangle. Then the
set of all solutions to $I$ can be found from $(S,\sigma)$ in
$O(|S_I|cd+1)$ time, where $S_I$ is the set of solutions to $I$.
\end{proposition}

\begin{proof}
In the trivial case in which $|S_I|=0$, we necessarily have as input
$S = \emptyset$ which can clearly be tested for in $O(1)$ time. In
another simple case, in which $I$ has at most two variables, the
result follows from arc consistency. We therefore only need to
consider satisfiable instances with at least three variables.

We now consider the elimination of a single variable $x$ from an
instance $I$ due to absence of one of the var-elim patterns BTP or
$\exists$subBTP. As observed in the proof of
Proposition~\ref{prop:one-solution}, each solution of the reduced
instance can be extended to a solution of $I$. This implies that the
number of solutions cannot decrease when we reinstate the variable
$x$. Clearly each solution of $I$ is an extension of a solution of
the reduced instance. So the algorithm given in the proof of
Proposition~\ref{prop:one-solution}, applied in turn to each solution
of the reduced instance will find all solutions of $I$ in time
$O(|S_I|c_x d)$.

Now consider the elimination of a value $b$ from the domain of a
variable $x$ by absence of $\exists$2triangle on values $a,b \in
\mathcal{D}(x)$ in an instance $I$. As observed in the proof of
Theorem~\ref{thm:all-val-elim -patterns}, $s$ is a solution to $I$
with $s(x) = b$ implies that $s''$ defined by $s''(x) = a$, $s''(v) =
s(v)$ for $v \neq x$ is a solution to the reduced instance. To
determine all solutions of $I$ from the set of all solutions of the
reduced instance thus requires only $O(|S_I|c_x)$ time.

Summing over all variables $x$ (and, in the case of
value-eliminations, over all assignments to $x$), we obtain a total
time complexity of $O(|S_I|cd+1)$, as claimed.
\end{proof}

An essential element of the proof of
Proposition~\ref{prop:all-solutions} in the case of var-elim patterns
is that the number of solutions does not decrease when a variable $x$
is reinstated. Unfortunately, in the case of the var-elim patterns
$\exists$invsubBTP and $\exists$snake, it is easy to construct an
example in which this is not true. Indeed, consider a binary CSP
instance $I$ corresponding to the 2-colouring of a star graph (a
graph composed of one central node with edges to $n-1$ other nodes).
Let $x$ be the variable corresponding to the central node of the
graph. Since neither $\exists$invsubBTP nor $\exists$snake occur on
$x$, both rules allow us to eliminate $x$, leaving an instance on
$n-1$ variables and no constraints. Whereas $I$ has only two
solutions (corresponding to the two possible 2-colourings of a star
graph), the reduced instance has $2^{n-1}$ solutions. In this
example, reinstating a single variable decreased the number of
solutions by an exponential factor.

The elimination of values can, on the other hand, dramatically
simplify an instance to the extent that finding all solutions to the
original instance remains intractable, as illustrated by the
following proposition.

\begin{proposition}
Let $I$ be an arc-consistent binary CSP instance and suppose that we
are given the set of all solutions to the instance obtained after
applying a single value-elimination operation due to the absence of
one of the patterns $\exists$2snake or $\exists$2invsubBTP.
Determining whether $I$ has more than one solution is NP-complete.
\end{proposition}

\begin{proof}
The problem is clearly in NP. It therefore suffices to give a
polynomial reduction from the known NP-complete problem binary CSP.
Let $J = \tuple{X, D, A, \cpt}$ be an arbitrary instance of binary
CSP. We will build an instance $I_J$ such that, after elimination of
one variable $x$ from $I_J$ by either $\exists$2invsubBTP or
$\exists$2snake and re-establishing arc consistency, we obtain a
trivially-solvable instance with exactly one solution, but
determining the existence of a second solution to $I_J$ is equivalent
to solving the instance $J$.

The variable-set of instance $I_J$ is $X \cup \{x\}$ (where $x$ is a
variable not in $X$). For each variable $y \in X$, the domain of $y$
in $I_J$ is $\mathcal{D}(y) \cup \{0\}$, where without loss of
generality we assume that $0$ does not belong to the domain
$\mathcal{D}(y)$ of variable $y$ in $J$. The domain of variable $x$
in $I_J$ is $\{0,1\}$. The compatibility function of $I_J$ is an
extension of the compatibility function of $J$: for each $y \in X$,
the assignment $\tuple{x,0}$ is compatible only with the assignment
$\tuple{y,0}$, whereas the assignment $\tuple{x,1}$ is compatible
with all the assignments $\tuple{y,a}$ for $a \neq 0$; furthermore
for each $y,z \in X$, the assignment $\tuple{y,0}$ is compatible with
all assignments to $z$.

Neither $\exists$2invsubBTP nor $\exists$2snake occur on variable $x$
in $I_J$ with $a,b$ mapping respectively to $0,1$. We can therefore
eliminate the value $1$ from the domain of $x$. After establishing
arc consistency, all domains are reduced to the singleton $\{0\}$.
Hence the reduced instance has exactly one solution. In the instance
$I_J$, the assignment $\tuple{x,0}$ only belongs to the solution
assigning 0 to each variable, whereas the assignment $\tuple{x,1}$ is
compatible with exactly the set of solutions to the instance $J$.
Therefore, determining the existence of a second solution to $I_J$ is
equivalent to determining the satisfiability of $J$.
\end{proof}

\section{Conclusion}

This paper has introduced the notion of
variable and value elimination rules in binary CSPs based on the
absence of quantified patterns. We have identified all irreducible
quantified patterns whose absence allows variable or value
elimination. As a consequence, we have also identified novel
tractable classes of binary CSPs. From a practical point of view, our
rules can be incorporated into generic constraint solvers to prune
the search tree.

There are several interesting directions for further research.
Can we generalise the variable or value elimination patterns
described in this paper to arbitrary-arity CSP instances
(perhaps using one of the possible definitions of microstructure for constraints
of arbitrary arity~\cite{DBLP:conf/sara/MouelhiJT13})?
A partial positive answer to this question has recently been
provided by arbitrary-arity versions of BTP~\cite{DBLP:conf/cp/CooperMTZ14}.
Do the variable and value elimination patterns introduced in this
paper generalise to other versions of constraint satisfaction, such
as the QCSP (as is the case for the tractable class defined by
BTP~\cite{GaoYZ11}) or the Weighted CSP (as is the case for tractable
class defined by the so-called joint-winner
pattern~\cite{Cooper12:tractable})? The research reported in the present paper
has recently led to the discovery of sound variable and value
elimination rules defined by local properties which strictly generalise the absence
of patterns~\cite{DBLP:conf/cp/Cooper14}. The characterisation of all
such generalised variable or value elimination rules is a challenging open problem.


\end{document}